\documentclass[10pt,a4paper]{article}

\usepackage{amsmath,amssymb,amsfonts,amsthm}
\usepackage[utf8]{inputenc}
\usepackage{graphicx}
\usepackage{subfig}
\usepackage{bbm}
\usepackage{verbatim}
\usepackage[pdftex,bookmarks=false,colorlinks=true,linkcolor=blue,
citecolor=blue,filecolor=black,urlcolor=blue]{hyperref}

\theoremstyle{plain} 
\newtheorem{definition}{Definition}

\newtheorem{proposition}[definition]{Proposition}

\newtheorem{corollary}[definition]{Corollary}

\numberwithin{equation}{section}

\begin{document}

\title{Matrix-valued Boltzmann Equation for the Hubbard Chain}

\author{Martin L.R. F\"urst$^\text{1,2,a}$, Christian B. Mendl$^\text{1,b}$, Herbert Spohn$^\text{1,3,c}$}

\footnotetext[1]{Mathematics Department, Technische Universit\"at M\"unchen, Boltzmannstra{\ss}e 3, 85748 Garching, Germany}
\footnotetext[2]{Excellence Cluster Universe, Technische Universit\"at M\"unchen, Boltzmannstra{\ss}e 2, 85748 Garching bei M\"unchen, Germany}
\footnotetext[3]{Physics Department, Technische Universit\"at M\"unchen, James-Franck-Stra{\ss}e 1, 85748 Garching bei M\"unchen, Germany}

\renewcommand*{\thefootnote}{\alph{footnote}}

\footnotetext[1]{\href{mailto:mfuerst@ma.tum.de}{mfuerst@ma.tum.de},
$^\text{b}$\href{mailto:mendl@ma.tum.de}{mendl@ma.tum.de},
$^\text{c}$\href{mailto:spohn@ma.tum.de}{spohn@ma.tum.de}}

\renewcommand*{\thefootnote}{\arabic{footnote}}

\date{\today}

\maketitle

\begin{abstract}
\noindent We study, both analytically and numerically, the Boltzmann transport equation for the Hubbard chain with nearest neighbor hopping and spatially homogeneous initial condition. The time-dependent Wigner function is matrix-valued because of spin. The H-theorem holds. The nearest neighbor chain is integrable which, on the kinetic level, is reflected by infinitely many additional conservation laws and linked to the fact that there are also non-thermal stationary states. We characterize all stationary solutions. Numerically, we observe an exponentially fast convergence to stationarity and investigate the convergence rate in dependence on the initial conditions.
\end{abstract}

\tableofcontents

\section{Introduction}

The Hubbard model is a simplified description of interacting electrons moving in a periodic background potential, see~\cite{Essler2005,Fehske2008,Rasetti1991} for introductory literature. We are interested in the dynamics of the Hubbard model in the regime of small interactions, which is conveniently described by kinetic theory, following the pioneering work of Peierls~\cite{Peierls1929}, Nordheim~\cite{Nordheim1928}, and Uehling, Uhlenbeck~\cite{UehlingUhlenbeck1933}. From the point of view of kinetic theory the Hubbard model has unusual features. The non-interacting model has a doubly degenerate band, which -- because of spin -- makes the Wigner function $2 \times 2$ matrix-valued. In addition the hamiltonian is invariant under a global SU(2) spin rotations. On the kinetic level this property is reflected by an exceptionally large set of conserved quantities. We refer to~\cite{FermionicTransport2012} for a recent experimental realization through ultracold atoms in an optical lattice under conditions where also kinetic theory is applied. 

As one would expect, even the matrix-valued Boltzmann equation satisfies the H-theorem. The goal of our note is to achieve -- beyond mere entropy increase -- a quantitative and more detailed understanding of the approach to stationarity. The Boltzmann equation consists of the sum of an effective hamiltonian plus a dissipative collision term, both with cubic nonlinearity. At such generality, numerical simulation is not an easy task. Therefore we concentrate on the Hubbard chain with nearest neighbor hopping and on-site interaction. In addition we assume spatial homogeneity. Our simulations use $64$ grid points in momentum space, which still allows for easy exploration. At this stage the reader might wonder, why on the kinetic level in one dimension there are any collisions at all. This will be explained in due course, as well as the difference between the nearest neighbor integrable model and the non-integrable next nearest neighbor case.

Let us start with the underlying hamiltonian and the resulting kinetic equation. The electrons are described by a spin-$\tfrac{1}{2}$ Fermi field on $\mathbb{Z}$ with creation/annihilation operators satisfying the anticommutation relations
\begin{equation}
\{ a_\sigma^*(x), a_\tau(y) \} = \delta_{xy} \delta_{\sigma \tau}, \qquad \{ a_\sigma(x), a_\tau(y) \} = 0, \qquad \{ a_\sigma^*(x), a_\tau^*(y) \} = 0
\end{equation}
for $x, y \in \mathbb{Z}$, $\sigma, \tau \in \{\uparrow, \downarrow\}$, and $\{ A, B \} = A B + B A$. The hamiltonian reads
\begin{equation}
H = \sum_{x, y \in \mathbb{Z}} \alpha(x - y) \, a^*(x) \cdot a(y)
		+ \frac{\lambda}{2} \sum_{x \in \mathbb{Z}} \big( a^*(x) \cdot a(x) \big)^2.
\end{equation}
Here $a^*(x) \cdot a(x) = a_\uparrow^*(x) \, a_\uparrow(x) + a_\downarrow^*(x) \, a_\downarrow(x)$. $\alpha$ is the hopping amplitude, with the properties $\alpha(x) = \alpha(x)^*$, $\alpha(x) = \alpha(-x)$, and $\lambda$ is the strength of the on-site interaction. Our notation emphasizes the invariance under global spin rotations.

For the Fourier transformation we use the convention
\begin{equation}
\label{eq:Fouriertransformation}
\hat{f}(k) = \sum_{x \in \mathbb{Z}} f(x)\, \mathrm{e}^{-2 \pi i\, k \cdot x}.
\end{equation}
Then the first Brillouin zone is the interval $\mathbb{T} = [-\tfrac{1}{2}, \tfrac{1}{2}]$ with periodic boundary conditions. The dispersion relation $\omega(k) = \hat{\alpha}(k)$ and, up to a constant, in Fourier space $H$ can be written as
\begin{equation}
H = \sum_{\sigma \in \{\uparrow,\downarrow\}} \int_{\mathbb{T}} \mathrm{d} k \, \omega(k) \hat{a}^*_\sigma(k) \hat{a}_\sigma(k) 
	+ \lambda \int_{\mathbb{T}^4} \mathrm{d}^4 \boldsymbol{k} \, \delta(\underline{k}) \, 
		\hat{a}_\uparrow^*(k_1) \hat{a}_\uparrow^*(k_2) \hat{a}_\downarrow(k_3) \hat{a}_\downarrow(k_4)
\end{equation}
with $\underline{k} = k_1 + k_2 - k_3 - k_4 \mod 1$ and $\mathrm{d}^4 \boldsymbol{k} = \mathrm{d}k_1\, \mathrm{d}k_2\, \mathrm{d}k_3\, \mathrm{d}k_4$.

To arrive at the kinetic equation, we assume that the initial state of the chain is quasifree, gauge invariant, and invariant under spatial translations. It is thus completely characterized by the two-point function
\begin{equation} \label{eq:InitialCondition}
\langle \hat{a}_\sigma^*(k) \hat{a}_\tau(k') \rangle = \delta(k - k') W_{\sigma \tau}(k).
\end{equation}
It will be convenient to think of $W(k)$ as a $2 \times 2$ matrix for each $k \in \mathbb{T}$. Then, in general, $W(k_1) W(k_2) \neq W(k_2) W(k_1)$ and every argument of standard kinetic theory has to be reworked. By the Fermi property we have $0 \leq W(k) \leq 1$ as a matrix for each $k$. In particular, $W$ can be written as
\begin{equation}
W(k) = \sum_{\sigma \in \{\uparrow,\downarrow\}} \varepsilon_\sigma(k) \lvert k,\sigma\rangle \langle k,\sigma\rvert,
\end{equation}
where $\lvert k,\sigma\rangle$ for $\sigma \in \{\uparrow,\downarrow\}$ is a $k$-dependent basis in spin space $\mathbb{C}^2$ and $\varepsilon_\sigma$ are the eigenvalues with $0 \leq \varepsilon_\sigma \leq 1$.

At some later time $t$ the state is still gauge and translation invariant, hence necessarily
\begin{equation}
\langle a_\sigma^*(k,t) a_\tau(k',t) \rangle = \delta(k - k') W_{\sigma \tau}(k,t).	
\end{equation}
In general $W(t)$ is a complicated object, but for small coupling $\lambda$ the quasi-free property persists over a time scale of order $\lambda^{-2}$, a structure which allows one to obtain the kinetic equation by second order time-dependent perturbation theory. More details can be found, e.g., in~\cite{ErdosSalmhoferYau2004,NotToNormalOrder2009,MeiLukkarinenSpohnPrep}. Here we only write down the resulting Boltzmann equation
\begin{equation}
\label{eq:BoltzmannEquation}
\frac{\partial}{\partial t} W(k,t)
= \mathcal{C}_\mathrm{c}[W](k,t) + \mathcal{C}_\mathrm{d}[W](k,t) = \mathcal{C}[W](k,t),
\end{equation}
which has the structure of an evolution equation and has to be supplemented with the initial data $W(k,0) = W(k)$.

The first term is of Vlasov type,
\begin{equation}
\label{eq:Cc}
\mathcal{C}_\mathrm{c}[W](k,t) := - i\,[H_\mathrm{eff}(k,t), W(k,t)],
\end{equation}
where the effective hamiltonian $H_\mathrm{eff}(k,t)$ is a $2 \times 2$ matrix which itself depends on $W$. More explicitly,
\begin{multline}
\label{eq:Heff}
H_{\mathrm{eff},1} = \int_{\mathbb{T}^3} \mathrm{d}k_2 \mathrm{d}k_3 \mathrm{d}k_4 \, \delta(\underline{k}) \, \mathcal{P} \left(\tfrac{1}{\underline{\omega}}\right)\\
\times \big( W_3 W_4 - W_2 W_3 - W_3 W_2 - \mathrm{tr}[W_4] W_3 + \mathrm{tr}[W_2] W_3 + W_2 \big).
\end{multline}
Here and later on we use the shorthand $\tilde{W} = 1 - W$, $W_1 = W(k_1,t)$, $H_{\mathrm{eff},1} = H_\mathrm{eff}(k_1,t)$, $\underline{\omega} = \omega(k_1) + \omega(k_2) - \omega(k_3) - \omega(k_4)$. Since $W$ is $2 \times 2$ matrix-valued, $\mathrm{tr}[\,\cdot\,]$ is the trace in spin space. Finally $\mathcal{P}$ denotes the principal part. Since the $k_3$, $k_4$ integration can be interchanged, $H_{\mathrm{eff}} = H_{\mathrm{eff}}^*$, as it should be.

There are many different ways to write the collision term $\mathcal{C}_\mathrm{d}$. We choose a version which separates the various contributions into gain and loss term. Then
\begin{equation}
\label{eq:Cd}
\mathcal{C}_\mathrm{d}[W]_1 = \pi \int_{\mathbb{T}^3} \mathrm{d}k_2 \mathrm{d}k_3 \mathrm{d}k_4 \delta(\underline{k}) \delta(\underline{\omega}) \big( \mathcal{A}[W]_{1234} + \mathcal{A}[W]_{1234}^* \big),
\end{equation}
where the index $1234$ means that the matrix $\mathcal{A}[W]$ depends on $k_1$, $k_2$, $k_3$, and $k_4$. Explicitly
\begin{multline}
\mathcal{A}[W]_{1234} = - W_4 \tilde{W}_3 W_2 + W_4 \, \mathrm{tr}[ \tilde{W}_3 W_2 ] \\
- \big\{\tilde{W}_4 W_2 - \tilde{W}_4 W_3 - \tilde{W}_3 W_2 + \tilde{W}_4 \, \mathrm{tr}[ W_3 ] - \tilde{W}_4 \, \mathrm{tr}[ W_2 ] + \mathrm{tr}[ W_2 \tilde{W}_3 ] \big\} W_1
\end{multline}
with the first two summands the gain term and $\{ ... \} W_1$ the loss term. The gain term is always positive definite, as implied by the inequality
\begin{equation}
A \, \mathrm{tr}[B C] + C \, \mathrm{tr}[B A] - A B C - C B A \geq 0
\end{equation}
valid for arbitrary positive definite matrices $A, B, C$. Thus if an eigenvalue of $W(k,t)$ happens to vanish, the gain term pushes it back to values $> 0$. A similar argument can be made for $\tilde{W}(k,t)$, implying the propagation of the Fermi property~\cite{MeiLukkarinenSpohnPrep}, to say: if at $t = 0$ one has $0 \leq W(k) \leq 1$, then the solution to~\eqref{eq:BoltzmannEquation} also satisfies $0 \leq W(k,t) \leq 1$.

In our contribution, we report on a numerical solution of the kinetic equation~\eqref{eq:BoltzmannEquation}, emphasizing the approach to stationarity. To provide an outline, in Sec.~\ref{sec:Properties} we establish a few general properties of~\eqref{eq:BoltzmannEquation},~\eqref{eq:Heff},~\eqref{eq:Cd}. They hold for arbitrary $\omega$ and also for the obvious extension of~\eqref{eq:BoltzmannEquation} to $d$ dimensions. In particular, we show that the entropy production $\sigma[W] = \frac{d}{\mathrm{d}t} S[W]$ has the property $\sigma \geq 0$. The thermal state $W_\mathrm{FD}$ (the Fermi-Dirac distribution) satisfies $\mathcal{C}[W_\mathrm{FD}] = 0$ and hence also $\sigma[W_\mathrm{FD}] = 0$. But to list all stationary solutions of~\eqref{eq:BoltzmannEquation} is not an easy task in general.

In Sec.~\ref{sec:Hubbard1D} we restrict ourselves to the Hubbard chain with nearest neighbor hopping, i.e.,
\begin{equation}
\label{eq:omega}
\omega(k) = 1 - \cos(2 \pi k).	
\end{equation}
The first task is to discuss the kinematically allowed collisions, in other words the solutions of $\underline{\omega} = 0$ together with $\underline{k} = 0 \mod 1$. The nearest neighbor model has a special symmetry through which a large set of further stationary states, beyond the thermal ones, can be found. On the kinetic level, this reflects the integrability of the underlying quantum hamiltonian. In Sec.~\ref{sec:Numerics} our numerical procedure is outlined and in Sec.~\ref{sec:Results} it is used to study the dynamics for representative initial Wigner functions.

\section{General properties of the Hubbard kinetic equation}
\label{sec:Properties}

To emphasize generality, for this section only, we consider $\mathbb{Z}^d$ as underlying lattice. Hence $k_j \in \mathbb{T}^d$ with periodic boundary conditions. The SU(2) invariance of $H$ is reflected by 
\begin{equation}
\mathcal{C}[U^* W U] = U^* \mathcal{C}[W] U	
\end{equation}
for all $U \in \mathrm{SU(2)}$. Hence if $W(k,t)$ is a solution to~\eqref{eq:BoltzmannEquation}, so is $U^*\,W(k,t)\,U$. Also hermiticity is propagated in time, i.e., if $W(0) = W(0)^*$, then also $W(t) = W(t)^*$, which follows from
\begin{equation}
\mathcal{C}[W]^* = \mathcal{C}[W^*].
\end{equation}
Furthermore the Fermi property, $0 \leq W(t) \leq 1$, is propagated in time, see~\cite{MeiLukkarinenSpohnPrep} for details.

There are two conservation laws
\begin{itemize}
\item spin
\begin{equation}
\label{eq:SpinConservation}
\frac{\mathrm{d}}{\mathrm{d} t} \int_{\mathbb{T}^d} \mathrm{d}k \, W(k,t) = 0,
\end{equation}
\item energy
\begin{equation}
\label{eq:EnergyConservation}
\frac{\mathrm{d}}{\mathrm{d} t} \int_{\mathbb{T}^d} \mathrm{d}k \, \omega(k) \, \mathrm{tr}[W(k,t)] = 0.
\end{equation} 
\end{itemize}
The proof uses the symmetrization of the integrand. One can interchange the variables $k_1 \leftrightarrow k_2$, $k_3 \leftrightarrow k_4$ and also the pairs $\{k_1, k_2\} \leftrightarrow \{k_3,k_4\}$. For the energy, one then picks up the integrand $\underline{\omega}$ and hence the factor $\underline{\omega}\,\delta(\underline{\omega}) = 0$.

%\begin{proof} (sketch)
%The spin conservation follows from~\eqref{eq:BoltzmannEquation}
%\begin{equation}
%\frac{\mathrm{d}}{\mathrm{d} t} \int_{\mathbb{T}^d} \mathrm{d} k_1 \, W_1 = \int_{\mathbb{T}^d} \mathrm{d} k_1 \, \mathcal{C}[W]_1
%\end{equation}
%and suitable interchange of $(k_1,k_2) \leftrightarrow (k_3,k_4)$, $k_1 \leftrightarrow k_2$ and $k_3 \leftrightarrow k_4$ within the resulting right integral over $\mathrm{d} k_1 \mathrm{d} k_2 \mathrm{d} k_3 \mathrm{d} k_4$. Similarly for the energy conservation,
%\begin{equation}
%\frac{\mathrm{d}}{\mathrm{d} t} \int_{\mathbb{T}^d} \mathrm{d} k_1 \, \omega_1 \, \mathrm{tr}[W_1] = \int_{\mathbb{T}^d} \mathrm{d} k_1 \, \omega_1 \, \mathrm{tr}[\mathcal{C}[W]_1]
%\end{equation}
%and interchange as for the spin conservation generates the form $(\omega_1 + \omega_2 - \omega_3 - \omega_4) \delta(\omega_1 + \omega_2 - \omega_3 - \omega_4) = 0$.
%\end{proof}

The next general property is the H-theorem. Since $\lvert\lambda\rvert \ll 1$, locally the state is free fermion. On the kinetic level, the entropy of the state $W$ is then defined as
\begin{equation}
S[W] = - \int_{\mathbb{T}^d} \mathrm{d} k_1 \big( \mathrm{tr}[W_1 \log W_1] + \mathrm{tr}[\tilde{W}_1 \log \tilde{W}_1] \big).
\end{equation}
Hence the entropy production is given by
\begin{equation}
\label{eq:sigmaW}
\sigma[W] = \frac{\mathrm{d}}{\mathrm{d} t} S[W] = -\int_{\mathbb{T}^d} \mathrm{d}k_1 \, \mathrm{tr}[ (\log W_1 - \log \tilde{W}_1) \, \mathcal{C}[W]_1 ].
\end{equation}
The H-theorem asserts that
\begin{equation}
\label{eq:HTheorem}
\sigma[W] \geq 0 \qquad \text{for all } W \text{ with } 0 \leq W \leq 1.   
\end{equation}
To establish~\eqref{eq:HTheorem}, for each $k$ we write
\begin{equation}
\label{eq:WEigenDecomp}
W(k) = \sum_{\sigma \in \{ \uparrow, \downarrow \} } \varepsilon_\sigma(k) P_\sigma(k)	
\end{equation}
with eigenvalues $0 \leq \varepsilon_\sigma(k) \leq 1$ and orthogonal eigen-projections $P_\sigma(k) = \lvert k,\sigma \rangle \langle k,\sigma\rvert$ with $\langle k,\sigma | k, \sigma' \rangle = \delta_{\sigma \sigma'}$. As before, we use a shorthand as $P_j = P_{\sigma_j}(k_j)$, $\varepsilon_j = \varepsilon_{\sigma_j}(k_j)$ and $\sum_{\boldsymbol{\sigma}} = \sum_{\sigma_1, \sigma_2, \sigma_3, \sigma_4}$. Inserting~\eqref{eq:WEigenDecomp} into~\eqref{eq:sigmaW}, one obtains
\begin{equation}
\begin{split}
\sigma[W] &= \pi \int_{(\mathbb{T}^d)^4} \mathrm{d}^4 \boldsymbol{k} \, \delta(\underline{k}) \delta(\underline{\omega}) \sum_{\boldsymbol{\sigma}}
	\big( \log\varepsilon_1 - \log\tilde{\varepsilon}_1 \big) 
	\big( \tilde{\varepsilon}_1 \tilde{\varepsilon}_2 \varepsilon_3 \varepsilon_4 - \varepsilon_1 \varepsilon_2 \tilde{\varepsilon}_3 \tilde{\varepsilon}_4 \big)\\
& \quad \times \big( \mathrm{tr}[P_1 P_3] \mathrm{tr}[P_2 P_4] + \mathrm{tr}[P_1 P_3] \mathrm{tr}[P_2 P_4] - \mathrm{tr}[P_1 P_3 P_2 P_4] - \mathrm{tr}[P_4 P_2 P_3 P_1] \big)\\
&= \pi \int_{(\mathbb{T}^d)^4} \mathrm{d}^4 \boldsymbol{k} \, \delta(\underline{k}) \delta(\underline{\omega}) \sum_{\boldsymbol{\sigma}}
	\big( \tilde{\varepsilon}_1 \tilde{\varepsilon}_2 \varepsilon_3 \varepsilon_4 - \varepsilon_1 \varepsilon_2 \tilde{\varepsilon}_3 \tilde{\varepsilon}_4 \big) \log(\tilde{\varepsilon_1}/\varepsilon_1)\\
& \quad \times \big\lvert \langle k_1,\sigma_1 | k_3,\sigma_3 \rangle \langle k_2,\sigma_2 | k_4,\sigma_4 \rangle - \langle k_1,\sigma_1 | k_4,\sigma_4 \rangle \langle k_2,\sigma_2 | k_3,\sigma_3 \rangle \big\rvert^2.
\end{split}
\end{equation}
We interchange $1 \leftrightarrow 2$, $3 \leftrightarrow 4$ and $(1,2) \leftrightarrow (3,4)$. Then
\begin{multline}
\label{eq:sigmaWfactorized}
\sigma[W]
= \frac{\pi}{4} \int_{(\mathbb{T}^d)^4} \mathrm{d}^4 \boldsymbol{k} \, \delta(\underline{k}) \delta(\underline{\omega}) \sum_{\boldsymbol{\sigma}} \left( \tilde{\varepsilon}_1 \tilde{\varepsilon}_2 \varepsilon_3 \varepsilon_4 - \varepsilon_1 \varepsilon_2 \tilde{\varepsilon}_3 \tilde{\varepsilon}_4 \right) \log\!\left( \frac{\tilde{\varepsilon}_1 \tilde{\varepsilon}_2 \varepsilon_3 \varepsilon_4}{\varepsilon_1 \varepsilon_2 \tilde{\varepsilon}_3 \tilde{\varepsilon}_4} \right) \\
\quad \times \big\lvert \langle k_1,\sigma_1 | k_3,\sigma_3 \rangle \langle k_2,\sigma_2 | k_4,\sigma_4 \rangle - \langle k_1,\sigma_1 | k_4,\sigma_4 \rangle \langle k_2,\sigma_2 | k_3,\sigma_3 \rangle \big\rvert^2
\geq 0,
\end{multline}
since $(x-y) \log(x/y) \geq 0$.

Stationary states are defined by
\begin{equation}
\mathcal{C}[W] = 0,
\end{equation}
which obviously implies $\sigma[W] = 0$. Physically one would expect thermal equilibrium to be included in the stationary states. On the kinetic level thermal equilibrium is defined by the Fermi-Dirac state
\begin{equation}
\label{eq:WFermiDirac}
W_\mathrm{FD}(k) = \sum_{\sigma \in \{\uparrow, \downarrow \}} \left(\mathrm{e}^{\beta (\omega(k) - \mu_\sigma)} + 1\right)^{-1} \lvert\sigma \rangle \langle \sigma\rvert,
\end{equation}
which is characterized by the inverse temperature $\beta$, the two chemical potentials $\mu_\uparrow$, $\mu_\downarrow$ for the spin occupations, and some $k$-independent spin basis $\lvert\sigma\rangle$, $\beta$, $\mu_\uparrow$, $\mu_\downarrow \in \mathbb{R}$. Indeed, it is easily checked that $\mathcal{C}[W_\mathrm{FD}] = 0$.

With this background information, one develops the following rough picture on the approach to stationarity. The initial state determines a special, $k$-independent basis $\lvert\sigma\rangle$ through
\begin{equation}
\label{eq:ConservedSigmaBasis}
\int_{\mathbb{T}^d} \mathrm{d}k \, W(k) = \sum_{\sigma \in \{\uparrow, \downarrow \}} \varepsilon_\sigma \, \lvert\sigma \rangle \langle \sigma\rvert.
\end{equation}
By~\eqref{eq:SpinConservation} this basis is preserved in time. Thus it is natural to expand $W(k,t)$ in this special basis. Approach to the thermal state would mean
\begin{equation}
\lim_{t \rightarrow \infty} \langle \sigma | W(k,t) | \sigma' \rangle = 0 \quad \text{for } \sigma \neq \sigma'
\end{equation}
and
\begin{equation}
\lim_{t \rightarrow \infty} \langle \sigma | W(k,t) | \sigma \rangle = \left(\mathrm{e}^{\beta (\omega(k) - \mu_\sigma)} + 1\right)^{-1} \quad \text{for } \sigma \in \{\uparrow, \downarrow \}.
\end{equation}
%
%where $\beta$, $\mu_\uparrow$, $\mu_\downarrow$ have to be determined from the initial $W$ through the conservation laws.
%
%The parameters $\beta$, $\mu_\uparrow$, $\mu_\downarrow$ are determined from the initial $W$ through the conservation %laws~\eqref{eq:SpinConservation} and~\eqref{eq:EnergyConservation} as
%
Since, by~\eqref{eq:SpinConservation}, the integral over the eigenvalue is conserved, one concludes that 
\begin{equation}
\varepsilon_\sigma = \int_{\mathbb{T}^d} \mathrm{d}k \left(\mathrm{e}^{\beta (\omega(k) - \mu_\sigma)} + 1\right)^{-1}, \quad \sigma \in \{\uparrow, \downarrow \}.
\end{equation}
Correspondingly, by~\eqref{eq:EnergyConservation}, for the average energy, 
\begin{equation}
\mathsf{e} = \int_{\mathbb{T}^d} \mathrm{d}k \, \omega(k) \, \mathrm{tr}[W(k)]
= \int_{\mathbb{T}^d} \mathrm{d}k \sum_{\sigma \in \{\uparrow, \downarrow \}} \omega(k) \left(\mathrm{e}^{\beta (\omega(k) - \mu_\sigma)} + 1\right)^{-1}.
\end{equation}
Both equations determine the parameters $\beta$, $\mu_\uparrow$, $\mu_\downarrow$ from the initial $W$. One has $0 \leq \varepsilon_\uparrow, \varepsilon_\downarrow \leq 1$ and $\min_k \omega(k) \leq \mathsf{e} \leq \max_k \omega(k)$. Then the map $(\mathsf{e},\varepsilon_\uparrow,\varepsilon_\downarrow)$ to $(\beta,\mu_\uparrow,\mu_\downarrow)$ is one-to-one.

Implicitly our argument assumes that the set of stationary states equals  the set of thermal states. But this might fail if there are not enough collisions, which could very well be the case in low dimensions. The issue of characterizing all stationary states has been accomplished only partially, see~\cite{CollisionalInvariants2006} for results towards this goal. On the other hand we still succeed in listing all stationary states and their domain of attraction. As to be discussed in the following section, for the Hubbard chain with nearest neighbor hopping the stationary states are not exhausted by the thermal ones.

\section{Nearest-neighbor Hubbard chain}
\label{sec:Hubbard1D}

\subsection{Collisions}
\label{sec:Collisions}
\begin{figure}[!ht]
\centering
\includegraphics[width=0.6\textwidth]{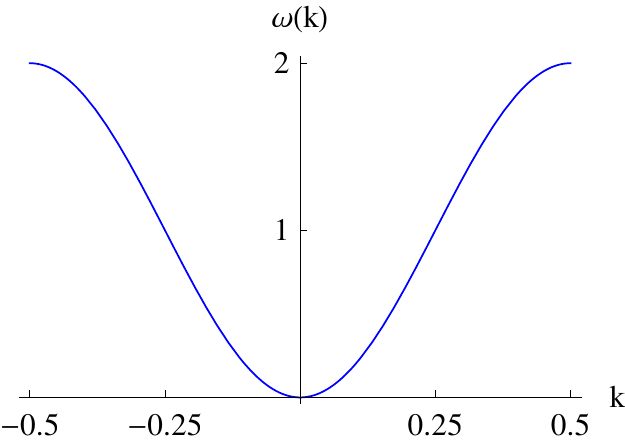}
\caption{The dispersion relation $\omega(k)$ of~\eqref{eq:omega}.}
\label{fig:Omega}
\end{figure}
We return to the Hubbard chain with nearest neighbor hopping~\eqref{eq:omega}. Fig.~\ref{fig:Omega} visualizes $\omega(k)$ for $k \in [-\tfrac12,\tfrac12]$. The first task is to investigate the kinematically allowed collisions defined by $\delta(\underline{k})\delta(\underline{\omega})$. The momentum conservation $\underline{k} = 0 \mod 1$ allows us to eliminate one $k$-variable, say $k_2 = k_3 + k_4 - k_1 \mod 1$. Inserted into energy conservation $\underline{\omega} = 0$ and using some trigonometric identities, one arrives at
\begin{equation}
\label{eq:EnergyConservationFactorized}
\underline{\omega} = 4 \sin(\pi(k_1-k_3)) \sin(\pi(k_1-k_4)) \cos(\pi(k_3+k_4)).
\end{equation}
Fig.~\ref{fig:OmegaEcons} visualizes $\underline{\omega}$ for fixed $k_1 = \tfrac{23}{64}$. From~\eqref{eq:EnergyConservationFactorized}, we conclude that the collision manifold has a solution path $k_3 + k_4 = \tfrac12$ (and thus also $k_1 + k_2 = \tfrac12$) denoted $\gamma_\mathrm{diag}$ in Fig.~\ref{fig:OmegaEcons}, besides the ``trivial'' solutions $k_3 = k_1$ (denoted $\gamma_1$) and $k_4 = k_1$ (denoted $\gamma_2$).

\begin{figure}[!ht]
\centering
\includegraphics[width=0.7\textwidth]{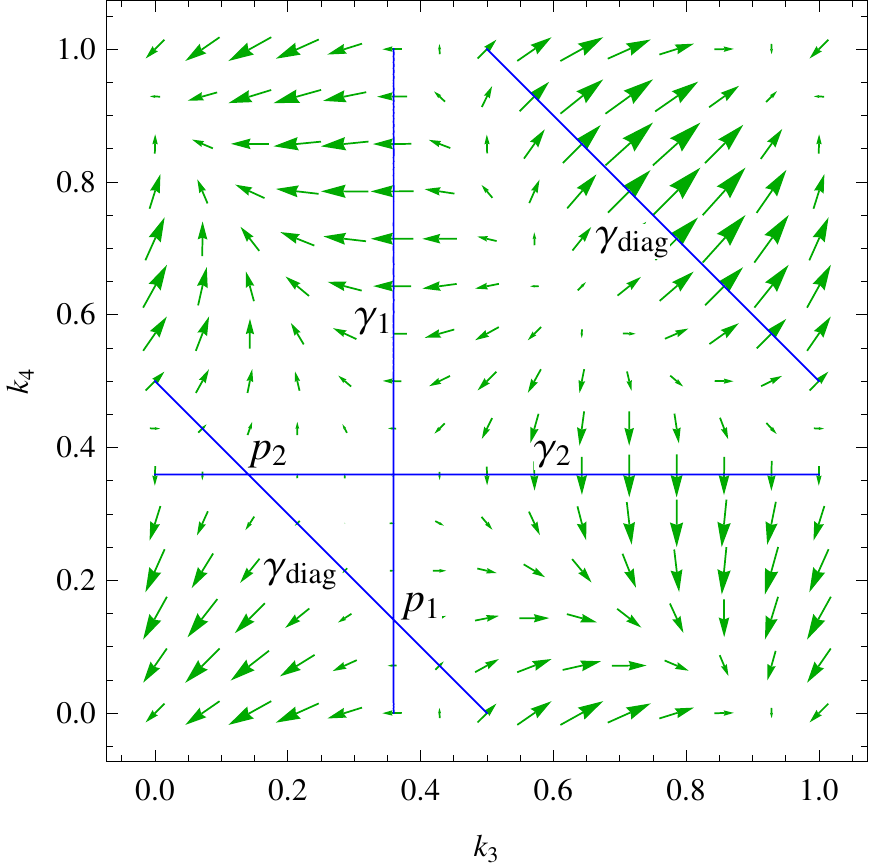}
\caption{Contour (blue straight lines) and gradient (green vectors) of the energy conservation $\underline{\omega} = 0$ for fixed $k_1 = \tfrac{23}{64}$ and after eliminating $k_2$. The diagonal blue line $\gamma_\mathrm{diag}$ is precisely the contour $k_3 + k_4 = \tfrac12$. The vertical and horizontal blue lines, $\gamma_1$ and $\gamma_2$, are the contours $k_3 = k_1$ and $k_4 = k_1$, respectively. $p_i$ marks the intersection of $\gamma_i$ with $\gamma_{\mathrm{diag}}$ for $i = 1, 2$.}
\label{fig:OmegaEcons}
\end{figure}

In what follows, we investigate the integral~\eqref{eq:Cd} of the dissipative collision operator $\mathcal{C}_{\mathrm{d}}$ along the paths $\gamma_1$, $\gamma_2$, and $\gamma_{\mathrm{diag}}$. Using the invariance of the integral~\eqref{eq:Cd} under $k_3 \leftrightarrow k_4$, we may interchange $W_3 \leftrightarrow W_4$. Then the integrand in~\eqref{eq:Cd} can be decomposed as
\begin{equation}
\mathcal{A}[W]_{1234} + \mathcal{A}[W]_{1234}^* = \mathcal{A}_{\mathrm{quad}}[W]_{1234} + \mathcal{A}_{\mathrm{tr}}[W]_{1234}
\end{equation}
with
\begin{align}
\mathcal{A}_{\mathrm{quad}}[W]_{1234} &:= -\tilde{W}_1 W_3 \tilde{W}_2 W_4 - W_4 \tilde{W}_2 W_3 \tilde{W}_1 + W_1 \tilde{W}_3 W_2 \tilde{W}_4 + \tilde{W}_4 W_2 \tilde{W}_3 W_1, \label{eq:Aquad}\\
\mathcal{A}_{\mathrm{tr}}[W]_{1234} &:= \big(\tilde{W}_1 W_3 + W_3 \tilde{W}_1\big) \mathrm{tr}[\tilde{W}_2 W_4] - \big(W_1 \tilde{W}_3 + \tilde{W}_3 W_1\big) \mathrm{tr}[W_2 \tilde{W}_4].
\label{eq:Atr}
\end{align}
Inspection of~\eqref{eq:Aquad} immediately reveals that $\mathcal{A}_{\mathrm{quad}}[W]_{1221} \equiv 0$ along $\gamma_2$ since $(k_1,k_2) = (k_4,k_3)$. Moreover, we also have $\mathcal{A}_{\mathrm{quad}}[W]_{1212} \equiv 0$ along $\gamma_1$ with $(k_1,k_2) = (k_3,k_4)$, which can be checked by expanding~\eqref{eq:Aquad}. In other words, $\mathcal{A}_{\mathrm{quad}}[W]_{1234}$ contributes only along $\gamma_{\mathrm{diag}}$.

The situation is different for the term $\mathcal{A}_{\mathrm{tr}}[W]_{1234}$: while $\mathcal{A}_{\mathrm{tr}}[W]_{1212} \equiv 0$ along $\gamma_1$ by direct inspection of~\eqref{eq:Atr}, it is (in general) non-zero along $\gamma_2$ and also along $\gamma_{\mathrm{diag}}$. In summary, for evaluating the dissipative collision integral~\eqref{eq:Cd} we have to integrate $\mathcal{A}_{\mathrm{tr}}[W]$ along $\gamma_2$ and both $\mathcal{A}_{\mathrm{quad}}[W]$ and $\mathcal{A}_{\mathrm{tr}}[W]$ along $\gamma_{\mathrm{diag}}$.

\medskip

As a side remark, the solution path $\gamma_{\mathrm{diag}}$ is special for the nearest neighbor dispersion relation~\eqref{eq:omega}. If we add to~\eqref{eq:omega} a small next-nearest neighbor term, then $\gamma_1$ and $\gamma_2$ persist and $\gamma_{\mathrm{diag}}$ gets somewhat deformed. In addition, a new collision channel opens up, as illustrated in Fig.~\ref{fig:OmegaNNNEcons} for the dispersion relation
\begin{equation}
\label{eq:omega_nnn}
\omega_{\mathrm{nnn}}(k) = \omega(k) - \frac12 \cos(4 \pi k) = 1 - \cos(2\pi k) - \frac12 \cos(4 \pi k).
\end{equation}

\begin{figure}[!ht]
\centering
\includegraphics[width=0.7\textwidth]{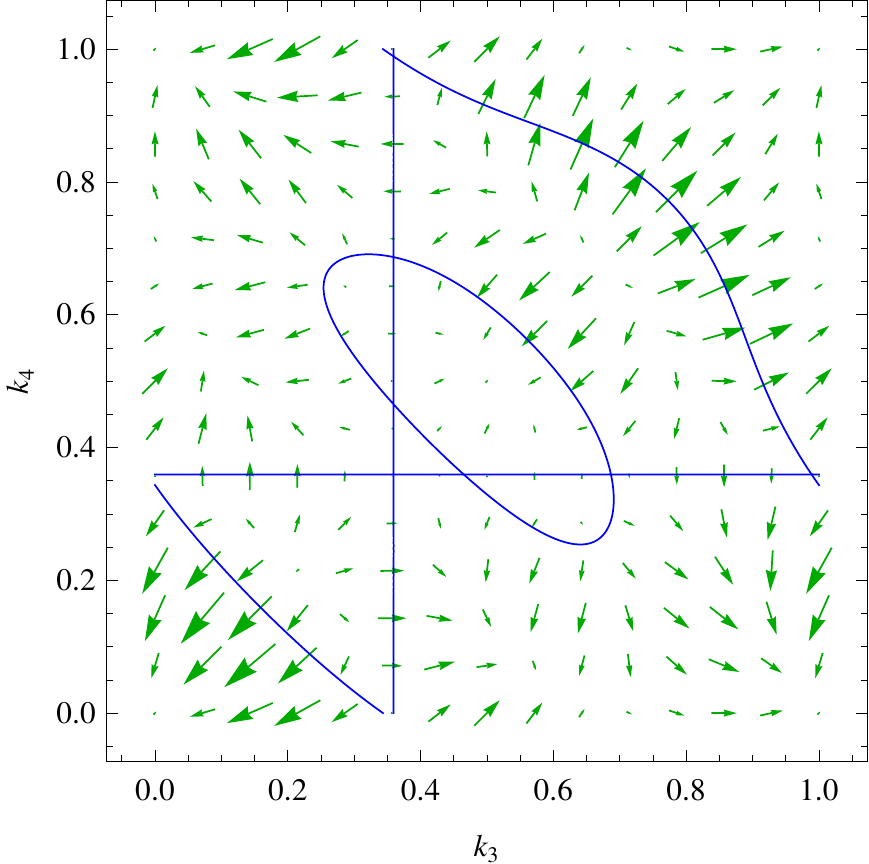}
\caption{Contour (blue straight lines) and gradient (green vectors) of the energy conservation for a next-nearest neighbor model with $\omega_{\mathrm{nnn}}(k)$ of~\eqref{eq:omega_nnn} and fixed $k_1 = \tfrac{23}{64}$.}
\label{fig:OmegaNNNEcons}
\end{figure}

\subsection{Stationary solutions}
\label{sec:StationarySolutions}

The collision paths $\gamma_1$, $\gamma_2$, and $\gamma_{\mathrm{diag}}$ have special symmetries, from which one can guess the form of stationary solutions beyond the thermal one. They have the same structure as the Fermi-Dirac state, but with $\omega(k)$ replaced by a more general function $f$. One finds
\begin{equation}
\label{eq:StationarySolutions}
W_{\mathrm{st}}(k) = \sum_{\sigma \in \{\uparrow,\downarrow\}} \lambda_\sigma(k) \, \lvert\sigma\rangle\langle\sigma\rvert, \quad \lambda_\sigma(k) = \left( \mathrm{e}^{f(k) - a_\sigma} + 1 \right)^{-1},
\end{equation}
where $f$ is a real-valued, $1$-periodic function satisfying $f(k) = -f(\tfrac12 - k)$, $a_\sigma \in \mathbb{R}$, and $\lvert\sigma\rangle$ is an orthogonal basis, independent of $k$.

As discussed in the Appendix,~\eqref{eq:StationarySolutions} characterizes the entire set of stationary solutions. The next step is to identify the domain of attraction for $W_{\mathrm{st}}$, in other word to study the map from the initial $W$ to $W_{\mathrm{st}}$. Here we can follow the strategy described at the end of Sec.~\ref{sec:Properties}. 

We first note that there are many energy-like quantities which are conserved. Let $g: \mathbb{T} \to \mathbb{R}$ with $g(k) = -g(\tfrac12 - k)$. Then
\begin{equation}
\label{eq:GeneralEnergyConservation}
\frac{\mathrm{d}}{\mathrm{d} t} \int_{\mathbb{T}} \mathrm{d}k\, g(k)\, \mathrm{tr}[W(k)] = 0,
\end{equation}
which generalizes the energy conservation~\eqref{eq:EnergyConservation}. \eqref{eq:GeneralEnergyConservation} again follows by an appropriate interchange of the integration variables $k_1,\dots,k_4$.

%
%hence
%
%\begin{equation} \label{eq:WST2}
%\int_{\mathbb{T}} \mathrm{d}k\, g(k)\, \mathrm{tr}[W(k)] = \int_{\mathbb{T}} %\mathrm{d}k\, g(k)\, \mathrm{tr}[W_{\mathrm{st}}(k)].
%\end{equation}
%
By substituting $g(k) = \delta(k - k') - \delta(k - \tfrac12 + k')$ for arbitrary $k' \in \mathbb{T}$, one concludes that
\begin{equation}
h(k) = \mathrm{tr}[W(k)] - \mathrm{tr}[W(\tfrac12 - k)]
\end{equation}
is pointwise constant for each $k \in \mathbb{T}$. Assuming that the initial $W$ converges to a stationary state of the form~\eqref{eq:StationarySolutions}, it must hold that
\begin{equation}
\label{eq:TraceDiffStationary}
h(k) = \sum_{\sigma \in \{\uparrow,\downarrow\}} \Big( \big( \mathrm{e}^{f(k) - a_\sigma} + 1 \big)^{-1} - \big( \mathrm{e}^{- f(k) - a_\sigma} + 1 \big)^{-1} \Big).
\end{equation}

Equivalently, as in Sec.~\ref{sec:Properties}, the spin conservation law requires that the eigenvalues $\varepsilon_\sigma$ in~\eqref{eq:ConservedSigmaBasis} are equal to
\begin{equation}
\label{eq:SpinStationary}
\varepsilon_\sigma = \int_\mathbb{T} \mathrm{d}k\, \left( \mathrm{e}^{f(k) - a_\sigma} + 1 \right)^{-1}.
\end{equation}
We claim that~\eqref{eq:TraceDiffStationary} and~\eqref{eq:SpinStationary} uniquely determine $f$ and $a_\sigma$, or more specifically, that the map between
\begin{equation} \label{eq:SET1}
\mathrm{tr}[W(k)] - \mathrm{tr}[W(\tfrac12 - k)], \lvert k \rvert \le \tfrac14, \qquad 0 \leq \varepsilon_\uparrow, \varepsilon_\downarrow \leq 1
\end{equation}
and
\begin{equation} \label{eq:SET2}
f(k) \text{ with } f(k) = -f(\tfrac12 - k), \lvert k \rvert \le \tfrac14, \qquad a_\uparrow, a_\downarrow
\end{equation}
is one-to-one. In particular, to a given $W$ one can associate a unique $W_\mathrm{st}$ of the form~\eqref{eq:StationarySolutions}.
\begin{proof}
By a short calculation,~\eqref{eq:TraceDiffStationary} can be written as
\begin{equation}
h(k) = - \sinh(f(k)) \left( \frac{1}{\cosh a_\uparrow  + \cosh f(k) } + \frac{1}{\cosh a_\downarrow  + \cosh f(k) } \right) \label{11}	
\end{equation}
and~\eqref{eq:SpinStationary} as
\begin{equation}
\varepsilon_\sigma = \int_\mathrm{I} \mathrm{d}k \left( \frac{\sinh a_\sigma }{\cosh a_\sigma  + \cosh f(k) } + 1 \right)
\end{equation}
with interval of integration $\mathrm{I} := [-\tfrac14,\tfrac14 ]$. We define a generalized ``free energy'' through
\begin{equation}
H(f,a_\uparrow,a_\downarrow) = \int_\mathrm{I} \mathrm{d}k \sum_{\sigma \in \{\uparrow,\downarrow\}} \log\big(\cosh a_\sigma + \cosh f(k) \big).
\end{equation}
The map $(f,a_\uparrow,a_\downarrow) \mapsto H$ is strictly convex. Furthermore
\begin{equation}
\frac{\partial}{\partial a_\sigma} H = \int_\mathrm{I} \mathrm{d}k \frac{\sinh a_\sigma }{\cosh a_\sigma + \cosh f(k) } = \varepsilon_\sigma - \frac{1}{2}	
\end{equation}
and
\begin{equation}
\frac{\delta H}{\delta f(k)} = \sum_{\sigma \in \{\uparrow,\downarrow\}} \frac{\sinh f(k) }{\cosh a_\sigma + \cosh f(k) } = - h(k).
\end{equation}
Thus the map from above can be viewed as Legendre transform from the first set~\eqref{eq:SET1} to the second set of variables~\eqref{eq:SET2}. Since $H$ is convex, the map is one-to-one.
\end{proof}

%
% \subsection{Spin}
%
%The operator to measure spin along an arbitrary axis direction is easily obtained from the Pauli spin matrices. Let $u = (u_x, u_y, u_z)$ be an arbitrary unit vector. Then the operator for spin in this direction is simply
%%
%\begin{equation}
%S_u = \frac{\hbar}{2} \left( u_x \sigma_x + u_y \sigma_y + u_z \sigma_z \right)	
%\end{equation}
%%
%The operator Su has eigenvalues of $\pm \hbar/2$, just like the usual spin matrices. This method of finding the operator for spin in an arbitrary direction generalizes to higher spin states, one takes the dot product of the direction with a vector of the three operators for the three x, y, z axis directions. 
%
%As well the Spin-distribution in the directions $\sigma_x$, $\sigma_y$, $\sigma_z$
%
%A normalized spinor for spin-$1/2$ in the $(u_x, u_y, u_z)$ direction (which works for all spin states except down where it will give 0/0), is:
%%
%\begin{equation}
%\phi(k) := \frac{1}{\sqrt{2+2 u_z(k)}} \left( \begin{array}{c} 1 + u_z(k) \\ u_x(k) + i u_y(k) \end{array} \right)	
%\end{equation}
%%
%and get
%%
%\begin{equation}
%W_0(k) := \phi(k) \phi(k)^T	
%\end{equation}
%%

\section{Numerical Procedure}
\label{sec:Numerics}

\subsection{Mollifying the collision operators}

\textit{Dissipative collision operator.} We have to make sure that $\delta(\underline{\omega}) \delta(\underline{k})$ is a well-defined prescription. For this purpose we eliminate $k_2$ and, using~\eqref{eq:EnergyConservationFactorized}, obtain

\begin{equation}
\int_{\gamma_2} \mathrm{d} k_4 \, \delta(\underline{\omega}) = \left\lvert\partial_{k_4}\underline{\omega} \vert_{k_4 = k_1}\right\rvert^{-1} = \left( 2\pi \left\lvert \sin(2\pi k_3) - \sin(2\pi k_1)\right\rvert \right)^{-1}.
\end{equation}
Likewise along $\gamma_{\mathrm{diag}}$ it holds that
\begin{equation}
\int_{\gamma_{\mathrm{diag}}} \mathrm{d} k_4 \, \delta(\underline{\omega}) = \left\lvert\partial_{k_4}\underline{\omega} \vert_{k_4 = 1/2 - k_3}\right\rvert^{-1} = \left( 2\pi \left\lvert \sin(2\pi k_3) - \sin(2\pi k_1)\right\rvert \right)^{-1}.
\end{equation}
Considering the subsequent integration over $k_3$ in~\eqref{eq:Cd}, the critical point $k_3 = \tfrac12 - k_1$ (marked $p_2$ in Fig.~\ref{fig:OmegaEcons}) would lead to infinities in general. (Integrating along $\gamma_\mathrm{diag}$ across the point $p_1$ ($k_3 = k_1$) is possible since $\mathcal{A}[W]_{1234} + \mathcal{A}[W]_{1234}^*$ is zero at that point, as explained above). As mollification we choose the substitution
\begin{equation}
\left(2\pi \left\lvert\sin(2\pi k_3) - \sin(2\pi k_1)\right\rvert\right)^{-1} \to \left( 4 \pi^2 \big( \sin(2\pi k_3) - \sin(2\pi k_1) \big)^2 + \epsilon^2 \right)^{-1/2}
\end{equation}
with some finite $\epsilon > 0$. Concretely, we use $\epsilon = \tfrac12$ for the simulations.

\medskip

\noindent \textit{Conservative collision operator.} The integral~\eqref{eq:Heff} for the conservative collision operator $\mathcal{C}_{\mathrm{c}}$ differs from dissipative integral~\eqref{eq:Cd}, since there is only a single delta distribution $\delta(\underline{k})$. Thus, we can eliminate $k_2 = k_3 + k_4 - k_1$ as for the dissipative case, but still have to integrate over both $k_3$ and $k_4$.

\begin{figure}[!ht]
\centering
\subfloat[]{
\includegraphics[width=0.5\textwidth]{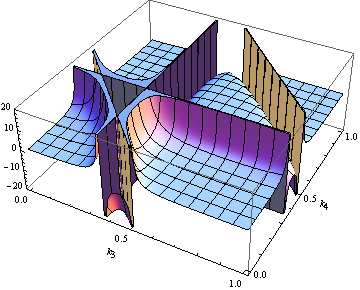}
\label{fig:OmegaInv}}
\subfloat[]{
\includegraphics[width=0.5\textwidth]{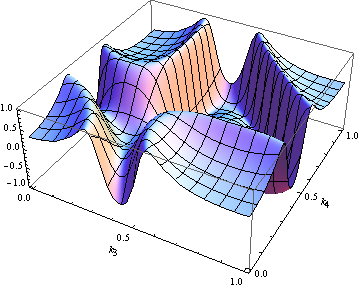}
\label{fig:OmegaInvMollified}}
\caption{(a) The term $1/\underline{\omega}$ as function of $k_3$ and $k_4$, for fixed $k_1 = 23/64$. (b) The ``mollified'' version $\underline{\omega}/(\underline{\omega}^2 + \epsilon^2)$ with $\epsilon = \frac12$ is free of singularities.}
\end{figure}

The integral~\eqref{eq:Heff} is defined as Cauchy principal value with respect to $1/\underline{\omega}$. Fig.~\ref{fig:OmegaInv} illustrates this term in dependence of $k_3$ and $k_4$ (compare also with Fig.~\ref{fig:OmegaEcons}). While the Cauchy principal value exists for continuous $W(k)$, the numeric calculation is rather demanding and we resort to a mollifying procedure as for the dissipative collision operator. Concretely, we substitute
\begin{equation}
\label{eq:InvOmegaMollified}
\frac{1}{\underline{\omega}} \to \frac{\underline{\omega}}{\underline{\omega}^2 + \epsilon^2}
\end{equation}
with finite $\epsilon > 0$ (in our case $\epsilon = \tfrac12$). Fig.~\ref{fig:OmegaInvMollified} shows the right-hand side, in direct comparison with the unmollified version. Note that $\mathcal{C}_{\mathrm{c}}$ could be defined via the integral~\eqref{eq:Heff} with the replacement~\eqref{eq:InvOmegaMollified}, and then letting $\epsilon \to 0$.

\subsection{Solving the Boltzmann equation}

In order to solve the Boltzmann equation~\eqref{eq:BoltzmannEquation} numerically, we discretize the $k$ variable by a uniform grid
\begin{equation}
k_j = \frac{j}{n}, \quad j = 0, \dots, n - 1
\end{equation}
with $n = 64$ in our case. We have chosen the interval $[0,1]$ instead of (equivalently) $[-1/2,1/2]$ simply for convenience. Note that due to periodicity, $W(1,t) = W(0,t)$, so the point $k = 1$ is not required. We use the trapezoidal rule to approximate the integrals~\eqref{eq:Cd} and~\eqref{eq:Heff} of the dissipative and conservative collision operators, respectively. Note that this approach is particularly suited for analytic period functions. Moreover, considering the 2-dimensional integral of the conservative collision operator, we ensure that the variable $k_2 = k_3 + k_4 - k_1 \mod 1$ is a grid point whenever $k_1$, $k_3$ and $k_4$ are grid points, in distinction from other integration rules with non-uniform points.

%For the numerical solution of~\eqref{} we use an operator splitting technique
%writing $W_t(k) \equiv W(k,t)$ and
We solve the Boltzmann differential equation~\eqref{eq:BoltzmannEquation} for the time variable by a \emph{Strang splitting} (or \emph{symmetric Trotter splitting}) technique: denoting the (fixed) timestep by $\Delta t$, we combine an explicit midpoint rule step for the dissipative part with the time evolution operator for the conservative part:
\begin{align}
\label{eq:timestepH1first}
&X(k_j,t) = \mathrm{e}^{-i H_\mathrm{eff}(k_j,t)\, \Delta t/2} \,W(k_j,t)\, \mathrm{e}^{i H_\mathrm{eff}(k_j,t)\, \Delta t/2}, \quad j = 0,\dots, n-1,\\
&Y(k_j,t) = X(k_j,t) + \Delta t\ \mathcal{C}_{\mathrm{d}}\!\left[X(t) + \frac{\Delta t}{2} \mathcal{C}_{\mathrm{d}}[X(t)]\right](k_j), \quad j = 0,\dots, n-1,\\
\label{eq:timestepH1second}
&W(k_j,t+\Delta t) = \mathrm{e}^{-i H'_\mathrm{eff}(k_j,t)\, \Delta t/2} \,Y(k_j,t)\, \mathrm{e}^{i H'_\mathrm{eff}(k_j,t)\, \Delta t/2}, \quad j = 0,\dots, n-1,
\end{align}
where $H'_\mathrm{eff}$ depends on $Y(t)$. The midpoint rule has order $2$, while the time evolution operator $\mathrm{e}^{-i H\, \Delta t/2} (\cdot) \mathrm{e}^{i H\, \Delta t/2}$ has only order $1$. Thus, the complete integration scheme has order $1$. As advantage, the time evolution operator preserves matrix symmetry. For the simulations, we use $\Delta t = 1/16$, and the overall simulation time interval runs from $t = 0$ to varying upper limit $t = 15, \dots, 45$.

\subsection{Cost analysis}

Considering a single time step, the most expensive part is the evaluation of the conservative collision operator $\mathcal{C}_{\mathrm{c}}$ in~\eqref{eq:timestepH1first} and~\eqref{eq:timestepH1second}, i.e., the 2-dimensional integral~\eqref{eq:Heff} after eliminating $k_2$. (The dissipative collision operator $\mathcal{C}_{\mathrm{d}}$ requires only a one-dimensional integration.) For the uniform discretization with $n$ points in each direction, this scales like $\mathcal{O}(n^2)$. One time step requires the evaluation of this integral for $n$ different $k_1$ points, thus the overall cost is $\mathcal{O}(n^3)$.

On a Intel Core i7-740QM Processor (6M cache, 1.73~GHz) without using parallelization, one time step takes approximately $90\,\mathrm{s}$ (Mathematica~8 implementation), so a complete simulation is approximately $6\,\mathrm{h}$. Note that the performance could be easily increased by a C/C++ implementation and making use of parallelization.

\section{Simulation results}
\label{sec:Results}

\textit{High-temperature state.} Fig.~\ref{fig:KMS_infiniteTmuW_FD} illustrates a high-temperature Fermi-Dirac equilibrium state $W_\mathrm{FD}$~\eqref{eq:WFermiDirac} where $\beta = 10^{-4}$, $\mu_\uparrow = 10^4$ and $\mu_\downarrow = -10^4$.
\begin{figure}[!ht]
\centering
\includegraphics[width=0.5\textwidth]{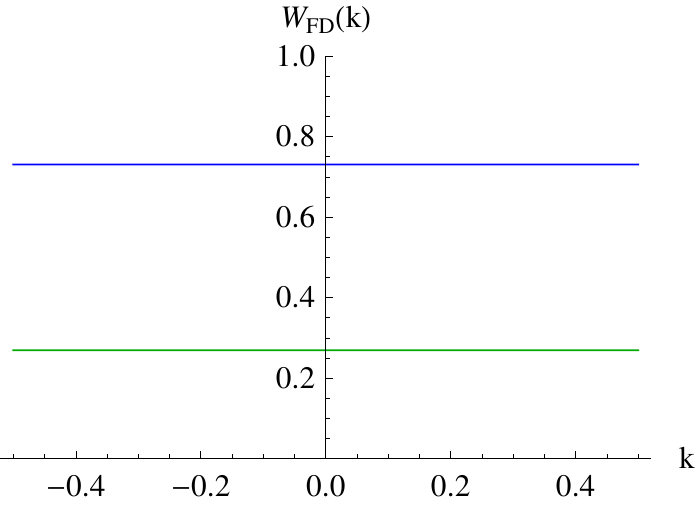}
\caption{Diagonal matrix entries of a high-temperature equilibrium state $W_\mathrm{FD}$. The state is (almost) independent of $k$.}
\label{fig:KMS_infiniteTmuW_FD}
\end{figure}
We have chosen the initial state
\begin{figure}[!ht]
\centering
\subfloat[matrix entries of initial $W(k,0)$]{
\includegraphics[width=0.5\textwidth]{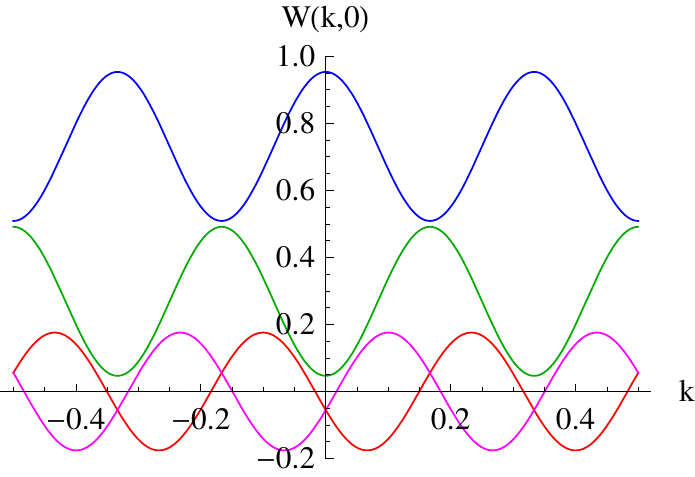}}
\subfloat[eigenvalues of $W(k,0)$]{
\includegraphics[width=0.5\textwidth]{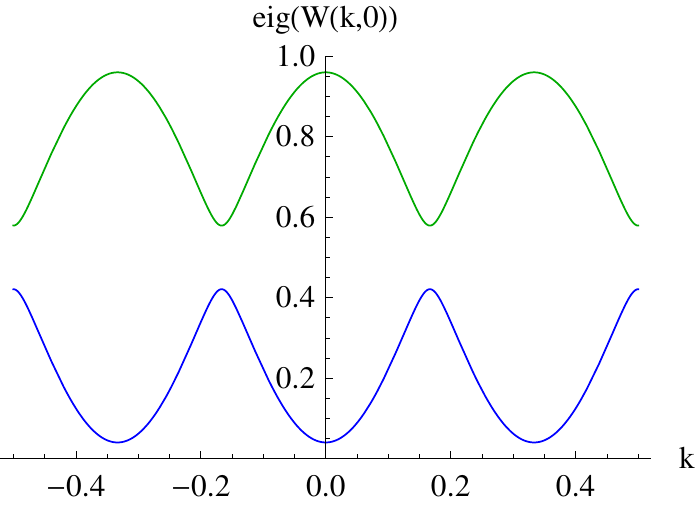}}
\caption{(a) Initial state $W(k,0) = W_\mathrm{FD}(k) + V(k)$ with $V(k)$ defined in~\eqref{eq:VinfiniteTmu}. The blue and green curves show the real diagonal entries, and the red and magenta curves the real and imaginary parts of the off-diagonal $\lvert\uparrow\rangle \langle\downarrow\rvert$ entry, respectively. (b) Corresponding eigenvalues of $W(k,0)$ in the interval $[0,1]$.}
\label{fig:KMS_infiniteTmuW0}
\end{figure}
(see Fig.~\ref{fig:KMS_infiniteTmuW0}) by $W(k,0) = W_\mathrm{FD}(k) + V(k)$ with $V(k)$ a rotation of the Pauli $\sigma_z$ matrix and subtracting the constant matrix $\tau/18$:
\begin{equation}
\label{eq:VinfiniteTmu}
V(k) = \frac{1}{4} \mathrm{e}^{-2 \pi i\,\tau\,k} \, \sigma_z \, \mathrm{e}^{2 \pi i\,\tau\,k} - \frac{\tau}{18}, \quad \tau = \sigma_x - \sigma_y + \frac12 \sigma_z.
\end{equation}
$V(k)$ satisfies
\begin{equation}
\int_{\mathbb{T}} \mathrm{d} k \, V(k) = 0  \quad \text{and} \quad \mathrm{tr}[V(k)] = 0,
\end{equation}
for all $k \in \mathbb{T}$ such that $W(0)$ matches $W_\mathrm{FD}$ in terms of the spin and energy conservation laws~\eqref{eq:SpinConservation} and~\eqref{eq:EnergyConservation}.

Fig.~\ref{fig:KMS_infiniteTmuConv} illustrates the convergence to the equilibrium state $W_\mathrm{FD}$. Interestingly, we observe that the off-diagonal entries converge \emph{slower} than the diagonal entries, but an analytic explanation of this effect is still lacking.
\begin{figure}[!ht]
\centering
\subfloat[convergence in Hilbert-Schmidt norm]{
\includegraphics[width=0.5\textwidth]{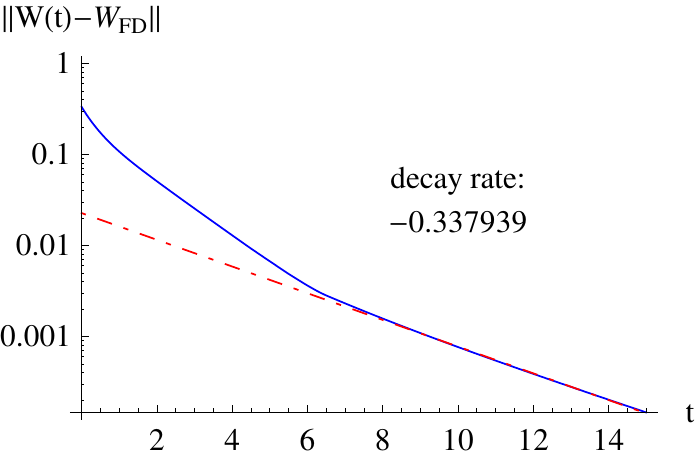}}
\subfloat[entropy convergence]{
\includegraphics[width=0.5\textwidth]{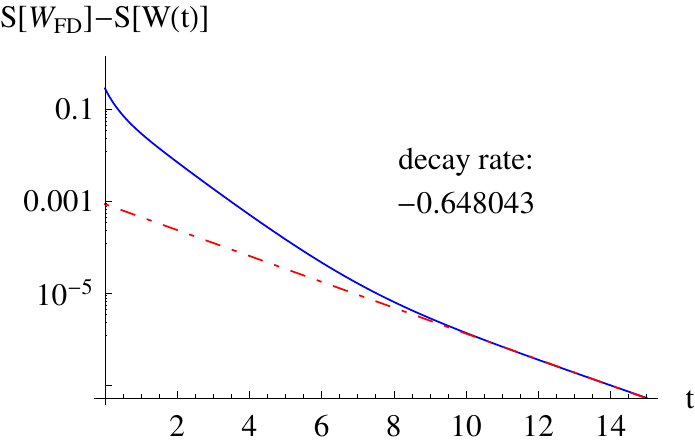}}\\
\subfloat[convergence of the diagonal entries]{
\includegraphics[width=0.5\textwidth]{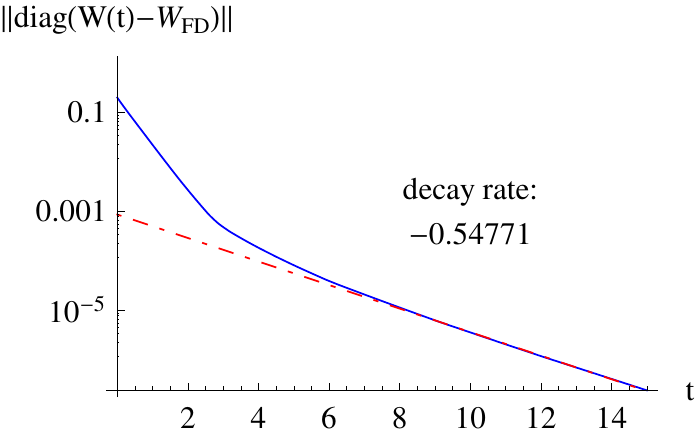}}
\subfloat[convergence of the off-diagonal entries]{
\includegraphics[width=0.5\textwidth]{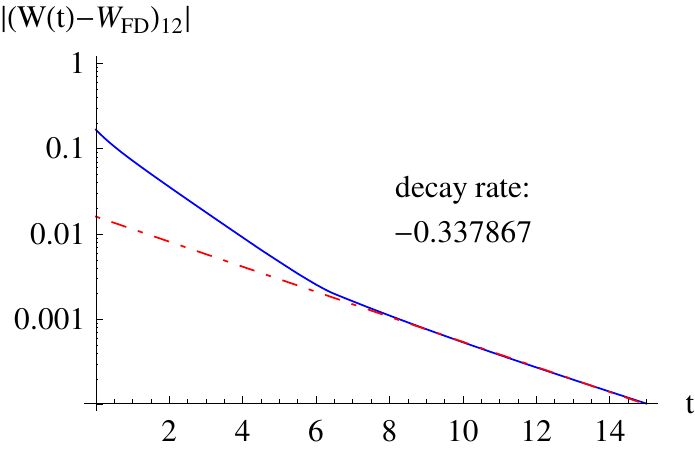}}
\caption{Convergence of the initial $W(0)$ (Fig.~\ref{fig:KMS_infiniteTmuW0}) to the high-temperature equilibrium state $W_\mathrm{FD}$ (Fig.~\ref{fig:KMS_infiniteTmuW_FD}) as semi-logarithmic plot (blue). The decay rate is the slope of the fitted red dotted line.}
\label{fig:KMS_infiniteTmuConv}
\end{figure}

Figure~\ref{fig:KMS_infiniteTmu_Bloch} displays the Bloch vectors $\vec{r}(k,t) \in \mathbb{R}^3$ of $W(k,t)$ parametrized by $k$, i.e.,
\begin{equation}
W(k,t) = \frac12 \left(\mathbbm{1} + \vec{r}(k,t) \cdot \vec{\sigma}\right), \quad \vec{\sigma} = (\sigma_x, \sigma_y, \sigma_z).
\end{equation}
The dark blue curve shows the initial $\vec{r}(k,0)$ and lighter blue curve shows $\vec{r}(k,\tfrac12)$. As time progresses, the initial curve straps to almost a single point, since $W_\mathrm{FD}(k)$ is almost independent of $k$.

\begin{figure}[!ht]
\centering
\subfloat[]{
\includegraphics[width=0.45\textwidth]{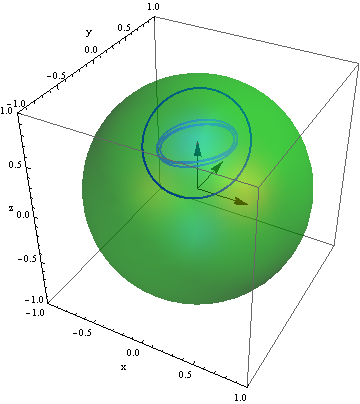}}
\subfloat[]{
\includegraphics[width=0.45\textwidth]{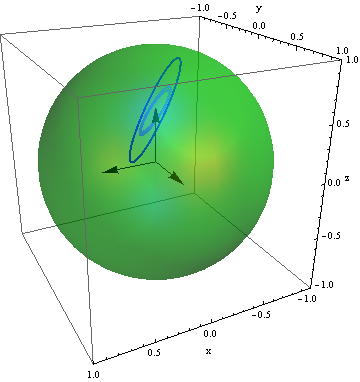}}
\caption{Bloch sphere representation (dark blue curve) of the initial $W(k,0)$ (Fig.~\ref{fig:KMS_infiniteTmuW0}), parametrized by $k$ and viewed from 2 perspectives. The light blue curve shows the corresponding Bloch curve of $W(k,t)$ for $t = 1/2$. Finally, the curve for the high-temperature equilibrium state $W_\mathrm{FD}(k)$ is indiscernible from a single point at the tip of the z-axis arrow.}
\label{fig:KMS_infiniteTmu_Bloch}
\end{figure}

\medskip

\textit{Low-temperature state.} Fig.~\ref{fig:KMS_lowT_FD} illustrates a low-temperature Fermi-Dirac equilibrium state $W_\mathrm{FD}$~\eqref{eq:WFermiDirac} with $\beta = 7$, $\mu_\uparrow = \tfrac{17}{16}$ and $\mu_\downarrow = \tfrac{15}{16}$.
\begin{figure}[!ht]
\centering
\includegraphics[width=0.5\textwidth]{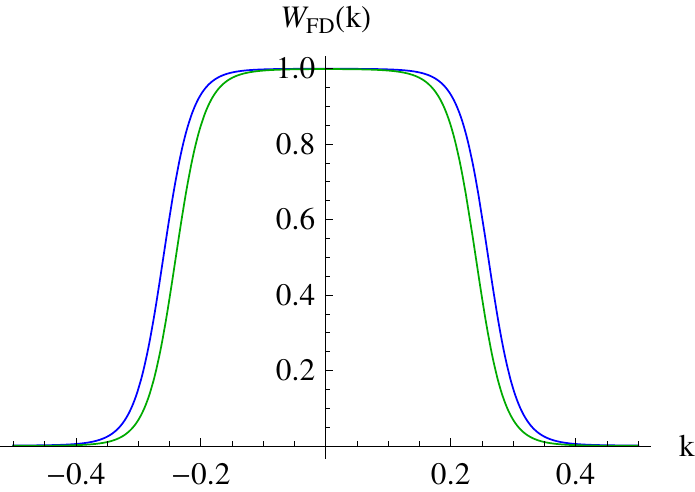}
\caption{Diagonal matrix entries of a low-temperature equilibrium state $W_\mathrm{FD}$ ($\beta = 7$, $\mu_\uparrow = \tfrac{17}{16}$, $\mu_\downarrow = \tfrac{15}{16}$).}
\label{fig:KMS_lowT_FD}
\end{figure}
In this case, for given $W_\mathrm{FD}(k)$ the variational freedom for the initial $W(k,0)$ with the same symmetries as $W_\mathrm{FD}(k)$ is strongly restricted. Similar to the high-temperature state, we define $W(k,0) = W_\mathrm{FD}(k) + V(k)$ (see Fig.~\ref{fig:KMS_lowT_W0})
\begin{figure}[!ht]
\centering
\subfloat[initial $W(k,0)$]{
\includegraphics[width=0.5\textwidth]{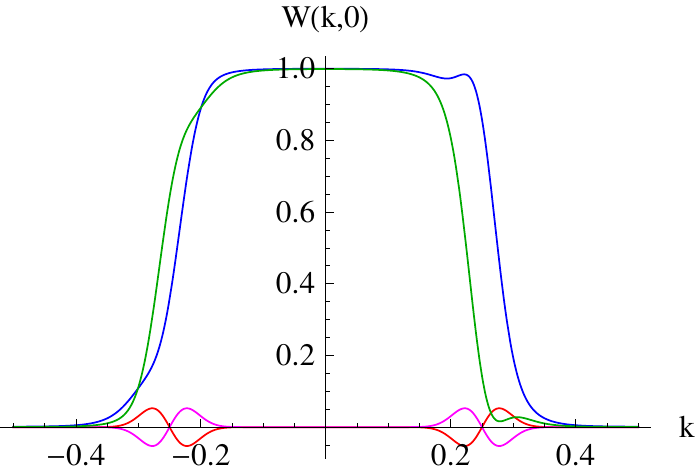}}
\subfloat[eigenvalues of $W(k,0)$]{
\includegraphics[width=0.5\textwidth]{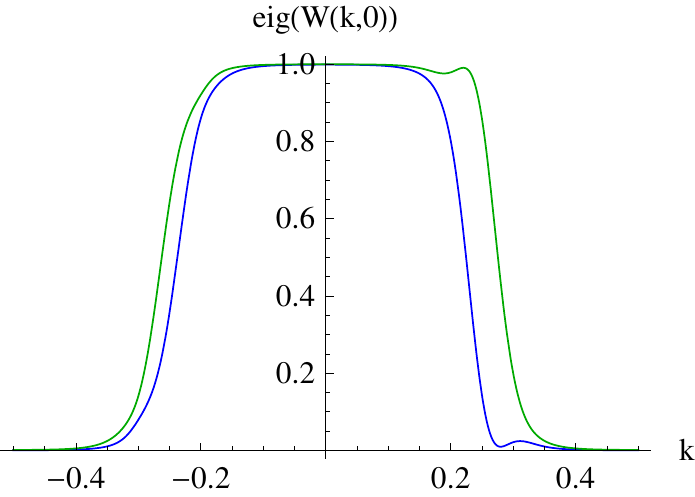}}
\caption{(a) Initial state $W(k,0) = W_\mathrm{FD}(k) + V(k)$ with $V(k)$ defined in~\eqref{eq:VlowT}. The blue and green curves show the real diagonal entries, and the red and magenta curves the purely imaginary off-diagonal entries. (b) Eigenvalues of $W(k,0)$ in the interval $[0,1]$.}
\label{fig:KMS_lowT_W0}
\end{figure}
with
\begin{equation}
\label{eq:VlowT}
V(k) = \frac14 \left(\mathrm{e}^{-64 \sin(\pi (k-3/4))^2}-\mathrm{e}^{-64 \sin(\pi (k-1/4))^2}\right) \mathrm{e}^{-2 \pi i\,\sigma_x\,k} \, \sigma_z \, \mathrm{e}^{2 \pi i\,\sigma_x\,k}.
\end{equation}
Again, $V(k)$ satisfies
\begin{equation}
\int_{\mathbb{T}} \mathrm{d} k \, V(k) = 0  \quad \text{and} \quad \mathrm{tr}[V(k)] = 0,
\end{equation}
for all $k \in \mathbb{T}$. We observe that the convergence to the equilibrium state (Fig.~\ref{fig:KMS_lowTConv}) is slower than for the high-temperature state in the previous paragraph. (Note that the simulation time interval is now $[0,45]$ as compared to $[0,15]$.)

\begin{figure}[!ht]
\centering
\subfloat[convergence in Hilbert-Schmidt norm]{
\includegraphics[width=0.5\textwidth]{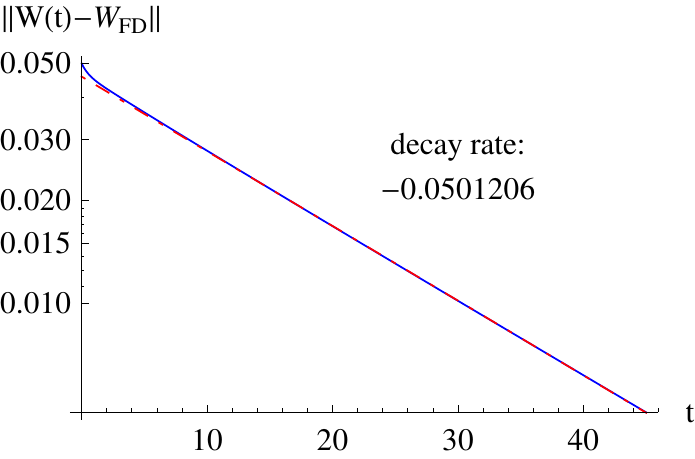}}
\subfloat[entropy convergence]{
\includegraphics[width=0.5\textwidth]{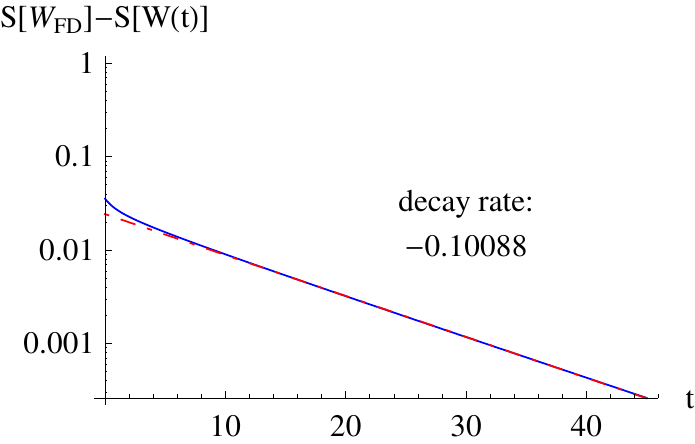}}\\
\subfloat[convergence of the diagonal entries]{
\includegraphics[width=0.5\textwidth]{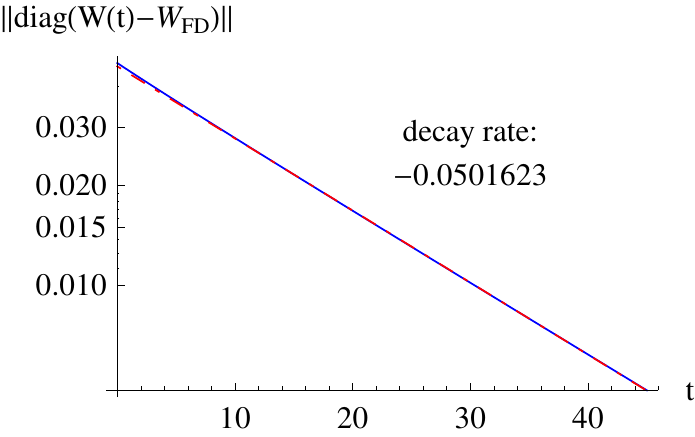}}
\subfloat[convergence of the off-diagonal entries]{
\includegraphics[width=0.5\textwidth]{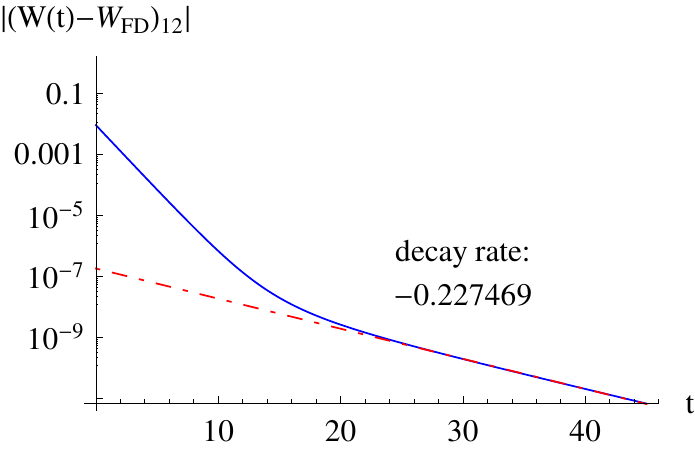}}
\caption{Convergence to the low-temperature equilibrium state $W_\mathrm{FD}$ (Fig.~\ref{fig:KMS_lowT_FD}) as semi-logarithmic plot (blue). The decay rate is the slope of the fitted red dotted line. For this example, we observe that the off-diagonal entries converge much faster than the diagonal ones.}
\label{fig:KMS_lowTConv}
\end{figure}

\medskip

\textit{Degenerate chemical potentials.} We consider a Fermi-Dirac equilibrium state with degenerate chemical potentials $\mu_\uparrow = \mu_\downarrow$, as illustrated in Fig.~\ref{fig:KMS_middleT_same_mu_FD}.
\begin{figure}[!ht]
\centering
\includegraphics[width=0.5\textwidth]{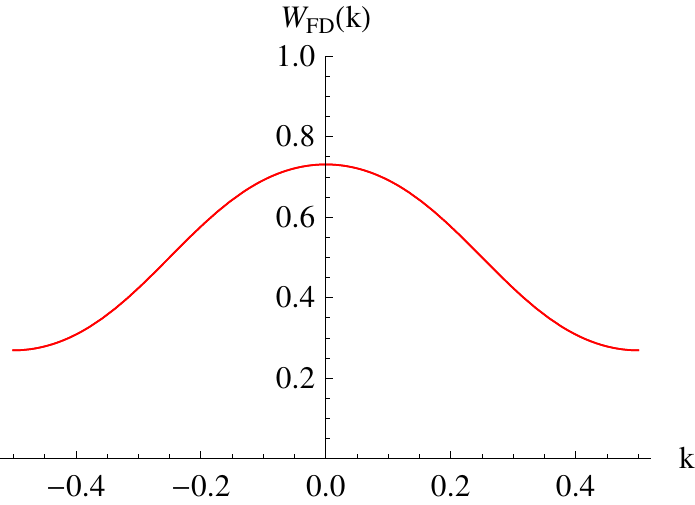}
\caption{Diagonal matrix entries of an equilibrium state $W_\mathrm{FD}$ with $\beta = 1$ and same chemical potentials $\mu_\uparrow = \mu_\downarrow = 1$.}
\label{fig:KMS_middleT_same_mu_FD}
\end{figure}
As initial state $W(k,0)$, we set $W(k,0) = W_\mathrm{FD}(k) + V(k)$ (see Fig.~\ref{fig:KMS_middleT_same_mu_W0}) with $V(k)$ taken from~\eqref{eq:VinfiniteTmu}.
\begin{figure}[!ht]
\centering
\subfloat[initial $W(k,0)$]{
\includegraphics[width=0.5\textwidth]{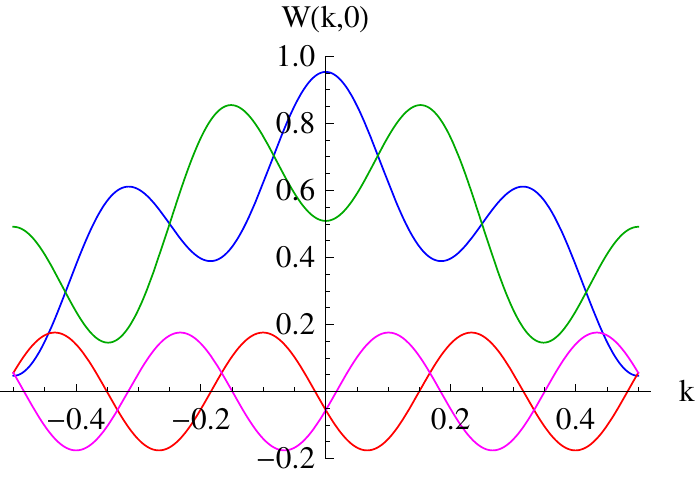}}
\subfloat[eigenvalues of $W(k,0)$]{
\includegraphics[width=0.5\textwidth]{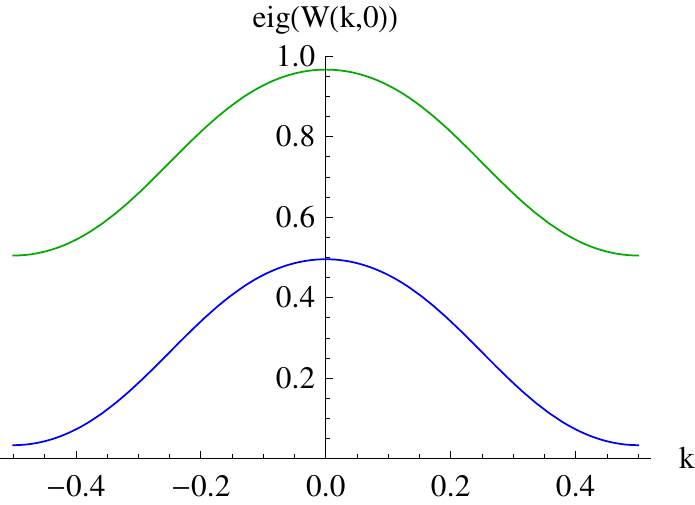}}
\caption{(a) Initial state $W(k,0) = W_\mathrm{FD}(k) + V(k)$ with $W_\mathrm{FD}(k)$ proportional to the identity matrix (see Fig.~\ref{fig:KMS_middleT_same_mu_FD}), and $V(k)$ defined in~\eqref{eq:VinfiniteTmu}. The blue and green curves show the real diagonal entries, and the red and magenta curves the purely imaginary off-diagonal entries. (b) The eigenvalues of $W(k,0)$ are non-degenerate, different from the equilibrium state $W_\mathrm{FD}(k)$.}
\label{fig:KMS_middleT_same_mu_W0}
\end{figure}
As illustrated in Fig.~\ref{fig:KMS_middleT_same_muConv}, there is no indication that the convergence changes due to the degeneracy.
\begin{figure}[!ht]
\centering
\subfloat[convergence in Hilbert-Schmidt norm]{
\includegraphics[width=0.5\textwidth]{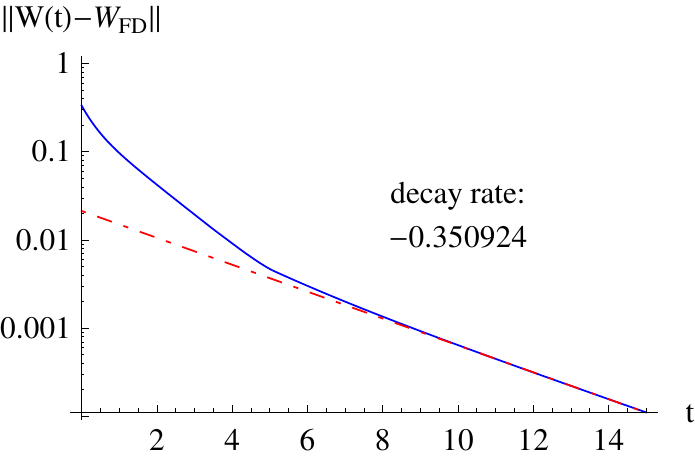}}
\subfloat[entropy convergence]{
\includegraphics[width=0.5\textwidth]{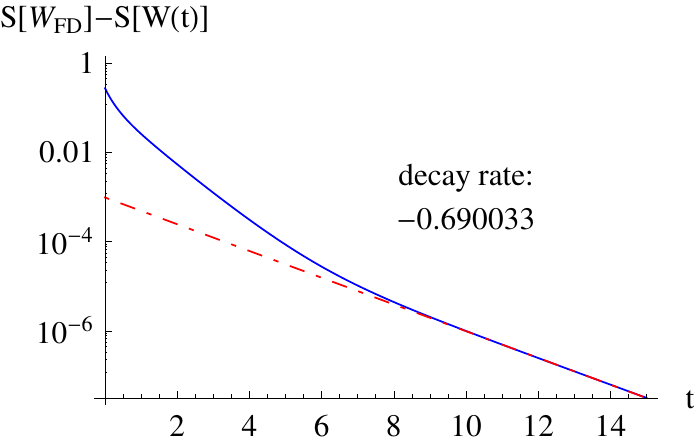}}
\caption{Convergence to the equilibrium state $W_\mathrm{FD}$ with degenerate eigenvalues (Fig.~\ref{fig:KMS_middleT_same_mu_FD}) as semi-logarithmic plot (blue). The decay rate is the slope of the fitted red dotted line.}
\label{fig:KMS_middleT_same_muConv}
\end{figure}
Two time snapshots of the eigenvalues of $W(k,t)$ are shown in Fig.~\ref{fig:KMS_middleT_same_muEigConv}. They have a peculiar shape, and converge to the diagonal entries of $W_\mathrm{FD}$, as expected.
\begin{figure}[!ht]
\centering
\subfloat[eigenvalues at $t = \tfrac12$]{
\includegraphics[width=0.5\textwidth]{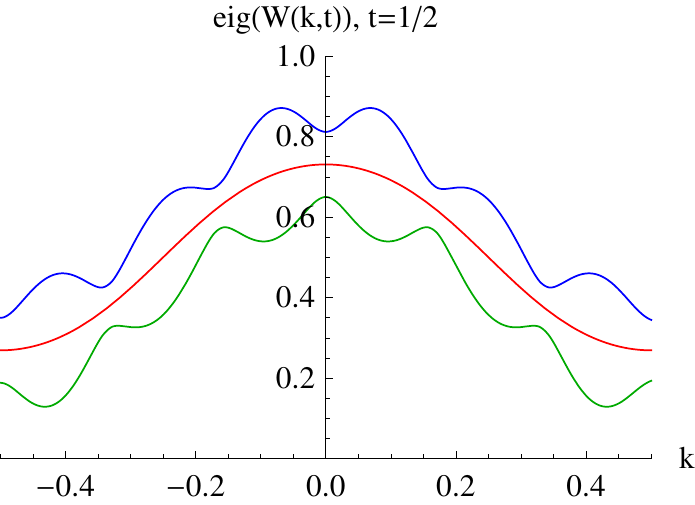}}
\subfloat[eigenvalues at $t = 2$]{
\includegraphics[width=0.5\textwidth]{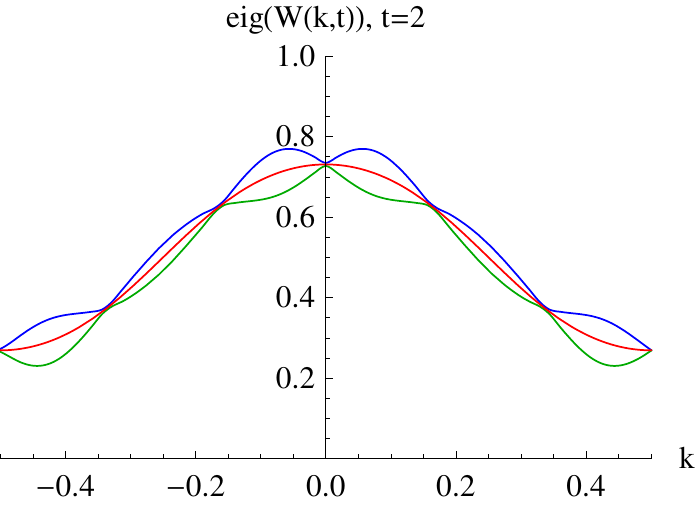}}
\caption{Two snapshots showing the convergence of the eigenvalues to the equilibrium state $W_\mathrm{FD}$ (red, same as Fig.~\ref{fig:KMS_middleT_same_mu_FD}) with $\mu_\uparrow = \mu_\downarrow = 1$.}
\label{fig:KMS_middleT_same_muEigConv}
\end{figure}

\medskip

\textit{Non-thermal stationary state.} For this example, we start from an (rather arbitrary) initial
\begin{multline}
\label{eq:initialW0nonThermal}
W(k,0) = \\
\frac{2}{5} \begin{pmatrix}
\frac12 \mathrm{e}^{-\cos(4\pi (k-\gamma))} + \frac14 & \frac14 \sin\!\left(\mathrm{e}^{2\pi i k}\right) \\ \frac14 \sin\!\left(\mathrm{e}^{-2\pi i k}\right) & \frac14 \mathrm{erf}(\cos(2\pi k)) + \frac12 + \arctan(\sin(2\pi k - \frac15)) + \frac{\pi}{4}\end{pmatrix}
\end{multline}
illustrated in Fig.~\ref{fig:stationary2W0} (where $\gamma$ is the Euler gamma constant), and then determine the stationary, non-thermal state $W_{\mathrm{st}}(k)$, via the $f$-function described in Sec.~\ref{sec:StationarySolutions}. Fig.~\ref{fig:stationary2fWst} illustrates both $f$ and $W_{\mathrm{st}}$. Next, we run the numerical simulation of the time evolution, which should converge to the predicted $W_{\mathrm{st}}(k)$. Fig.~\ref{fig:stationary2Conv} indeed verifies the convergence to $W_{\mathrm{st}}(k)$.

\begin{figure}[!ht]
\centering
\subfloat[initial $W(k,0)$]{
\includegraphics[width=0.5\textwidth]{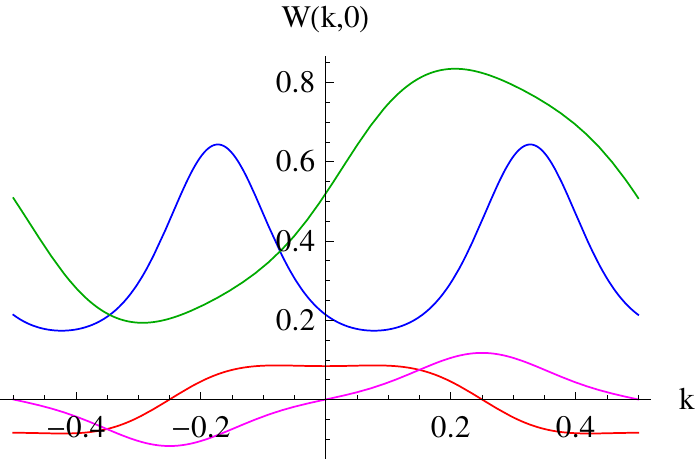}}
\subfloat[eigenvalues of $W(k,0)$]{
\includegraphics[width=0.5\textwidth]{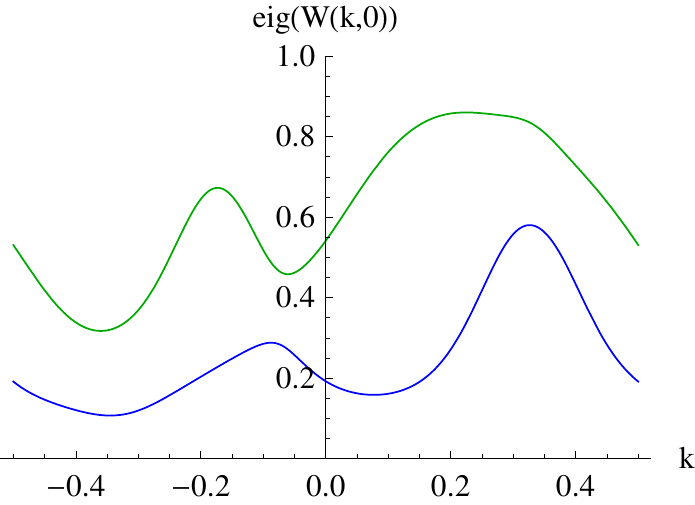}}
\caption{(a) Initial state $W(k,0)$ defined in~\eqref{eq:initialW0nonThermal}. The blue and green curves show the real diagonal entries, and the red and magenta curves the real and imaginary parts of the off-diagonal $\lvert\uparrow\rangle \langle\downarrow\rvert$ entry, respectively. (b) Eigenvalues of $W(k,0)$ in the interval $[0,1]$.}
\label{fig:stationary2W0}
\end{figure}

\begin{figure}[!ht]
\centering
\subfloat[]{
\includegraphics[width=0.5\textwidth]{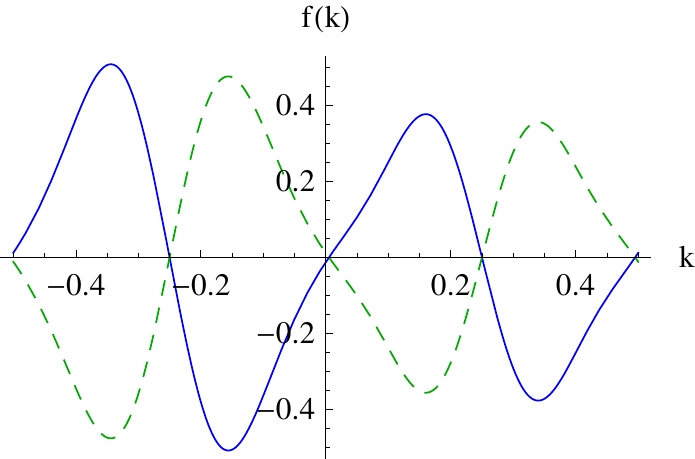}}
\subfloat[]{
\includegraphics[width=0.5\textwidth]{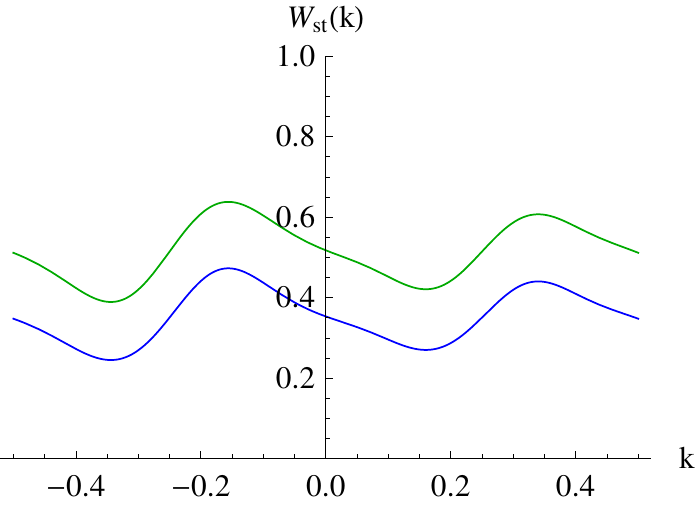}}\\
\caption{(a) The $f$-function (blue) calculated from $\mathrm{tr}[W(k,0)-W(\tfrac12-k,0)]$ (dashed) and the fitted ``chemical potentials'' $a_\uparrow = -0.617485$ and $a_\downarrow = 0.0578622$. The initial $W(k,0)$ is defined in~\eqref{eq:initialW0nonThermal}. (b) Resulting stationary state $W_\mathrm{st}(k)$~\eqref{eq:StationarySolutions} given by $f$ and $a_\uparrow$, $a_\downarrow$.}
\label{fig:stationary2fWst}
\end{figure}

\begin{figure}[!ht]
\centering
\subfloat[entropy convergence]{
\includegraphics[width=0.5\textwidth]{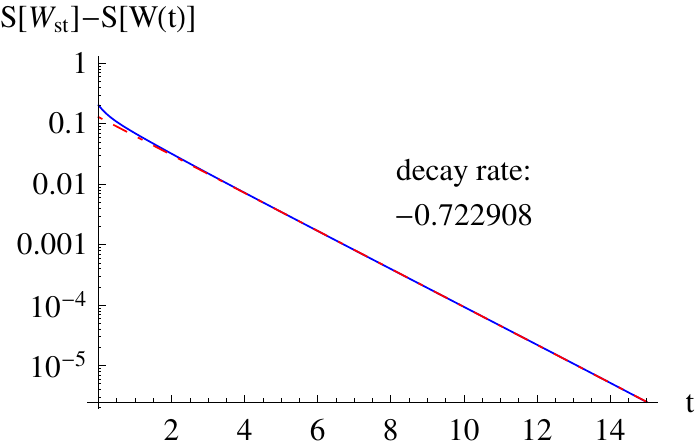}}
\subfloat[convergence in Hilbert-Schmidt norm]{
\includegraphics[width=0.5\textwidth]{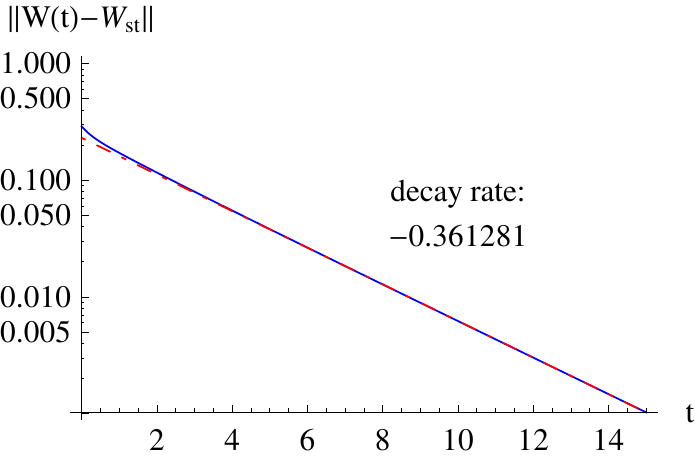}}
\caption{Convergence to the calculated $W_\mathrm{st}$ (Fig.~\ref{fig:stationary2fWst}) as semi-logarithmic plot.}
\label{fig:stationary2Conv}
\end{figure}

For special cases, we have checked that the asymptotic decay rate is almost independent of the initial conditions. This strongly suggests that the collision operator linearized at $W_\mathrm{st}$ has a spectral gap.

\clearpage

\section{Conclusions}
\label{sec:Conclusions}

The kinetic equation for the Hubbard model, in general, has two hardly investigated features (i) the Wigner function is $2 \times 2$ matrix-valued, (ii) the microscopic SU(2) invariance implies additional conservation laws. We investigated here the chain with nearest neighbor hopping, which is an integrable model, \cite{Essler2005}. The Boltzmann transport equation reflects integrability by an infinite number of conserved quantities and non-thermal stationary states. We established the H-theorem and classified all stationary states. Adding a next-nearest neighbor coupling seems to destroy all conservation laws beyond spin and energy which indicates that now the stationary solutions are exhausted by the thermal Fermi-Dirac Wigner functions.

In the spatially homogeneous case we observed numerically an exponentially fast convergence to the predicted stationary state, both for the diagonal and off-diagonal matrix elements with roughly comparable decay rates. The decay at low temperatures is slower than at high temperatures, as one would have expected. In principle, asymptotic decay rates can be computed from the linearized collision operator.

Physically of great interest would be to better understand the spatially inhomogeneous situation. For example one could imagine to have in each half of the chain a thermal state with the same temperature, but with different spin orientations. In principle, this could be handled by kinetic theory. One only would have to add in the kinetic equation the transport term $\omega'(k)\,\partial/\partial_x$. Numerically, such a problem is more demanding than the one studied here, but, at least in one dimension, still in reach. Another challenging problem would be to study energy transport through the chain. Our results point towards the validity of Fourier's law.

\appendix

\section{Characterization of stationary solutions}
\label{sec:AppendixStationary}

\begin{proposition}
\label{prop:sigmaWZero}
Let $\sigma[W]$ be as defined in~\eqref{eq:sigmaW}. If $0 < W < 1$, then the solutions to zero entropy production,
\begin{equation}
\sigma[W] = 0,
\end{equation}
are necessarily of the form~\eqref{eq:StationarySolutions}.
\end{proposition}
\noindent\textit{Remark.} As noted by J.~Lukkarinen, further zero entropy and stationary solutions are obtained by setting one eigenvalue of $W$ identically $= 0, 1$, and the other eigenvalue arbitrary.

\begin{proof}
On the one hand, if $W$ is of the form~\eqref{eq:StationarySolutions}, then $\sigma[W] = 0$ follows by inserting. On the other hand, let $\sigma[W] = 0$. We set $\boldsymbol{k} = (k_1,k_2,k_3,k_4)$, $\boldsymbol{\sigma} = (\sigma_1,\sigma_2,\sigma_3,\sigma_4)$, $\mathrm{d}^4 \boldsymbol{k} = \mathrm{d}k_1\, \mathrm{d}k_2\, \mathrm{d}k_3\, \mathrm{d}k_4$ and define
\begin{equation}
F(\boldsymbol{k},\boldsymbol{\sigma}) 
	= \left( \tilde{\varepsilon}_1 \tilde{\varepsilon}_2 \varepsilon_3 \varepsilon_4 - 
		\varepsilon_1 \varepsilon_2 \tilde{\varepsilon}_3 \tilde{\varepsilon}_4 \right) 
		\log \left( \frac{\tilde{\varepsilon}_1 \tilde{\varepsilon}_2 \varepsilon_3 \varepsilon_4}
			{\varepsilon_1 \varepsilon_2 \tilde{\varepsilon}_3 \tilde{\varepsilon}_4 } \right) \ge 0
\end{equation}
with the $\varepsilon_i$ as in Sec.~\ref{sec:Properties}, furthermore
\begin{equation}
\label{eq:Gdefinition}
G(\boldsymbol{k},\boldsymbol{\sigma}) 
	= \big\lvert \langle k_1,\sigma_1 | k_3,\sigma_3 \rangle \langle k_2,\sigma_2 | k_4,\sigma_4 \rangle 
		- \langle k_1,\sigma_1 | k_4,\sigma_4 \rangle \langle k_2,\sigma_2 | k_3,\sigma_3 \rangle \big\rvert^2.
\end{equation}
Then
\begin{equation}
\sigma[W] = \frac{\pi}{4} \int_{\mathbb{T}^4} \mathrm{d}^4 \boldsymbol{k} \, \delta(\underline{k}) \delta(\underline{\omega})
	\sum_{\boldsymbol{\sigma}} F(\boldsymbol{k},\boldsymbol{\sigma}) G(\boldsymbol{k},\boldsymbol{\sigma}) 	
\end{equation}
according to~\eqref{eq:sigmaWfactorized}. Since all terms are non-negative,
\begin{equation}
\label{eq:FGprodZero}
F(\boldsymbol{k},\boldsymbol{\sigma}) \, G(\boldsymbol{k},\boldsymbol{\sigma}) = 0
\end{equation}
must hold for all $\boldsymbol{\sigma}$ and all $\boldsymbol{k} \in \gamma_2 \cup \gamma_{\mathrm{diag}}$ (see Fig.~\ref{fig:OmegaEcons}). On $\gamma_1$ one has $F = 0$ and no extra information can be extracted.

$F$ has the structure $(x-y)\log(\tfrac{x}{y})$, which is zero only iff $x = y$, equivalently iff
\begin{equation}
\label{eq:LogDiff}
\log\left(\frac{\tilde{\varepsilon}_1 \tilde{\varepsilon}_2 \varepsilon_3 \varepsilon_4}
	{\varepsilon_1 \varepsilon_2 \tilde{\varepsilon}_3 \tilde{\varepsilon}_4}\right) = 
		\log\left( \frac{\tilde{\varepsilon}_1}{\varepsilon_1} \right)
		+ \log\left( \frac{\tilde{\varepsilon}_2}{\varepsilon_2} \right)
		- \log\left( \frac{\tilde{\varepsilon}_3}{\varepsilon_3} \right)
		- \log\left( \frac{\tilde{\varepsilon}_4}{\varepsilon_4} \right) = 0.
\end{equation}
Defining the \emph{collision invariants} as
\begin{equation} \label{eq:PhiDefinition}
\Phi_{\sigma}(k) = \log\left( \frac{\tilde{\varepsilon}_\sigma(k)}{\varepsilon_\sigma(k)} \right),
\end{equation}
condition~\eqref{eq:LogDiff} reads
\begin{equation} \label{eq:CollisionInvariants}
\Phi_{\sigma_1}(k_1) + \Phi_{\sigma_2}(k_2) = \Phi_{\sigma_3}(k_3) + \Phi_{\sigma_4}(k_4).
\end{equation}

Note that the labeling of eigenvalues $\varepsilon_\uparrow(k)$, $\varepsilon_\downarrow(k)$ and corresponding eigenvectors is arbitrary. Thus w.l.o.g.\ we can assume that
\begin{equation}
\label{eq:nperp}
\langle k_1,\uparrow |\, k_2,\uparrow \rangle \neq 0 \quad\text{and thus}\quad \langle k_1,\downarrow |\, k_2,\downarrow \rangle \neq 0
\end{equation}
for all $k_1, k_2 \in \mathbbm{T}$.

Consider the contour $\gamma_2$ ($k_1 = k_4$, $k_2 = k_3$) for $\boldsymbol{\sigma} = \uparrow\downarrow\uparrow\downarrow$. In this case, the second term on the right side of~\eqref{eq:Gdefinition} vanishes, and thus
\begin{equation}
G(\boldsymbol{k},\uparrow\downarrow\uparrow\downarrow) = \left\lvert \langle k_1,\uparrow |\, k_2,\uparrow\rangle \langle k_1,\downarrow |\, k_2,\downarrow\rangle \right\rvert^2 > 0
\end{equation}
by construction~\eqref{eq:nperp}. Therefore~\eqref{eq:FGprodZero} forces $F(\boldsymbol{k},\uparrow\downarrow\uparrow\downarrow) = 0$ on $\gamma_2$. Equation~\eqref{eq:CollisionInvariants} becomes after rearranging terms
\begin{equation}
\label{eq:FzeroSeparation}
\Phi_{\uparrow}(k_1) - \Phi_{\downarrow}(k_1) = \Phi_{\uparrow}(k_2) - \Phi_{\downarrow}(k_2).
\end{equation}
Since variables are separated, both sides of~\eqref{eq:FzeroSeparation} must be constant, i.e.,
\begin{equation} \label{eq:PhiUpDownConstant}
\Phi_{\uparrow}(k) - \Phi_{\downarrow}(k) = c
\end{equation}
for a fixed $c \in \mathbb{R}$ and all $k \in \mathbb{T}$.

Next, we establish that the basis $\lvert k,\sigma\rangle$ has to be $k$-independent up to a $k$-dependent phase, which can be chosen such that $\lvert k,\sigma\rangle = \lvert\sigma\rangle$ with $\lvert\uparrow\rangle, \lvert\downarrow\rangle$ a fixed basis in $\mathbb{C}^2$. If $c = 0$ in~\eqref{eq:PhiUpDownConstant}, then $\Phi_{\uparrow} = \Phi_{\downarrow}$, and it follows that $W(k) = \varepsilon(k)\,\mathbbm{1}$. In particular, one can set $\lvert k,\sigma\rangle = \lvert\sigma\rangle$. In the other case, $c \neq 0$, consider the contour $\gamma_2$ for $\boldsymbol{\sigma} = \uparrow\uparrow\downarrow\downarrow$:
\begin{equation}
F(\boldsymbol{k},\uparrow\uparrow\downarrow\downarrow) = \Phi_{\uparrow}(k_1) + \Phi_{\uparrow}(k_2) - \Phi_{\downarrow}(k_2) - \Phi_{\downarrow}(k_1) = 2 c \neq 0,
\end{equation}
where we have used~\eqref{eq:PhiUpDownConstant} for the second equality. Thus \eqref{eq:FGprodZero} requires that $G(\boldsymbol{k},\uparrow\uparrow\downarrow\downarrow) = 0$ on $\gamma_2$. Inserted into the definition~\eqref{eq:Gdefinition} yields
\begin{equation}
\label{eq:kProdPerp}
\langle k_1,\uparrow |\, k_2,\downarrow\rangle \langle k_2,\uparrow |\, k_1,\downarrow\rangle = 0
\end{equation}
for all $k_1, k_2 \in \mathbb{T}$. Since the vectors $\lvert k,\uparrow\rangle$ and $\lvert k,\downarrow\rangle$ are an orthonormal basis of $\mathbb{C}^2$ for each fixed $k$, \eqref{eq:kProdPerp} is equivalent to 
\begin{equation}
\label{eq:kPerp}
\langle k_1,\uparrow |\, k_2,\downarrow\rangle = 0
\end{equation}
for all $k_1, k_2 \in \mathbb{T}$. Keeping $k_2$ fixed, this means that $\lvert k_1,\uparrow\rangle = \mathrm{const}$ up to a phase, and similarly $\lvert k_1,\downarrow\rangle = \mathrm{const}$. W.l.o.g.\ the phase factor can be set to $1$, leaving invariant the projectors $P_\sigma(k) = \lvert k,\sigma\rangle\langle k,\sigma\rvert$. In summary, $\lvert k,\sigma\rangle = \lvert\sigma\rangle$ and $G(\boldsymbol{k},\boldsymbol{\sigma}) = G(\boldsymbol{\sigma})$.

As final step, consider $\gamma_\mathrm{diag}$ for $\boldsymbol{\sigma} = \uparrow\downarrow\uparrow\downarrow$. By direct inspection $G(\uparrow\downarrow\uparrow\downarrow) = 1$, thus \eqref{eq:FGprodZero} requires $F(\boldsymbol{k},\uparrow\downarrow\uparrow\downarrow) = 0$. \eqref{eq:CollisionInvariants} for $k_2 = \tfrac12 - k_1$ and $k_4 = \tfrac12 - k_3$ becomes
\begin{equation} \label{eq:PhiDiag}
\Phi_{\uparrow}(k_1) + \Phi_{\downarrow}(\tfrac12 - k_1) = \Phi_{\uparrow}(k_3) + \Phi_{\downarrow}(\tfrac12 - k_3).
\end{equation}
Since variables are separated, both sides must be constant, i.e.,
\begin{equation}
\Phi_{\uparrow}(k) + \Phi_{\downarrow}(\tfrac{1}{2}-k) = \mathrm{const}
\end{equation}
for all $k \in \mathbb{T}$. Combined with~\eqref{eq:PhiUpDownConstant}, we obtain
\begin{equation}
\Phi_{\sigma}(k) + \Phi_{\sigma}(\tfrac{1}{2}-k) = \mathrm{const}
\end{equation}
for $\sigma = \uparrow, \downarrow$. One concludes that $\Phi_\sigma$ is necessarily of the form
\begin{equation} \label{eq:Phi_fsigma}
\Phi_{\sigma}(k) = f_\sigma(k) - a_\sigma \quad \text{with} \quad f_\sigma(k) = -f_\sigma(\tfrac12 - k)
\end{equation}
for all $k \in \mathbb{T}$ and some $a_\sigma \in \mathbb{R}$. Plugging into~\eqref{eq:PhiUpDownConstant}, one deduces that $f_\uparrow(k) - f_\downarrow(k) = \mathrm{const}$ and, since $f_\sigma(\tfrac14) = 0$, it follows that $f_\uparrow(k) = f_\downarrow(k) = f(k)$ independent of $\sigma$. Summarizing, we arrive at
\begin{equation} \label{eq:Phi_f}
\Phi_{\sigma}(k) = f(k) - a_\sigma.
\end{equation}
Solving~\eqref{eq:Phi_f} and~\eqref{eq:PhiDefinition} for $\varepsilon_\sigma(k)$ leads to the claimed form~\eqref{eq:StationarySolutions}.
\end{proof}

\begin{corollary}
Under the constraint $0 < W < 1$, all stationary solutions, i.e., all solutions to $\mathcal{C}[W] = 0$, are precisely of the form~\eqref{eq:StationarySolutions},
\begin{equation}
\label{eq:StationarySolutions2}
W_{\mathrm{st}}(k) = \sum_{\sigma \in \{\uparrow,\downarrow\}} \lambda_\sigma(k) \, \lvert\sigma\rangle\langle\sigma\rvert, \quad \lambda_\sigma(k) = \left( \mathrm{e}^{f(k) - a_\sigma} + 1 \right)^{-1}
\end{equation}
with $f(k) = - f(\tfrac{1}{2} - k)$ for all $k \in \mathbb{T}$.
\end{corollary}
\begin{proof}
Each $W_\mathrm{st}$ of the form~\eqref{eq:StationarySolutions2} satisfies $\mathcal{C}[W] = 0$, which can be checked by inserting $W_\mathrm{st}$ into $\mathcal{C}[W]$: specifically, the commutator~\eqref{eq:Cc} defining $C_{\mathrm{c}}[W_\mathrm{st}]$ vanishes since $H_\mathrm{eff}$ and $W_\mathrm{st}$ are diagonal. The dissipative collision operator $C_{\mathrm{d}}[W_\mathrm{st}]$ is zero due to the symmetry properties of $\gamma_{\mathrm{diag}}$ and the fact that $f(k) = - f(\tfrac{1}{2} - k)$.
On the other hand, let $W$ be a solution of $\mathcal{C}[W] = 0$. Then
\begin{equation}
\frac{\partial}{\partial t} W(k,t) = \mathcal{C}[W](k,t) = 0,
\end{equation}
and, in particular,
\begin{equation}
\sigma[W] = \frac{\mathrm{d}}{\mathrm{d} t} S[W] = 0.
\end{equation}
According to Proposition~\ref{prop:sigmaWZero}, $W$ is of the form~\eqref{eq:StationarySolutions2}.
\end{proof}

%\bibliographystyle{plain}
%\bibliographystyle{unsrt}
%\bibliography{references}

\end{document}